\documentclass[]{infiniai}
\usepackage{microtype}
\usepackage{amsfonts}
\usepackage{graphicx}
\usepackage{booktabs}
\usepackage{wrapfig}
\usepackage{amsmath}
\usepackage{amsthm}
\usepackage{subcaption}

\usepackage{algpseudocode}
\usepackage[linesnumbered,lined,boxed,commentsnumbered,ruled,longend]{algorithm2e}
\usepackage[capitalize,noabbrev]{cleveref}
\theoremstyle{plain}
\newtheorem{theorem}{Theorem}[section]
\newtheorem{proposition}[theorem]{Proposition}
\newtheorem{lemma}[theorem]{Lemma}

\theoremstyle{definition}

\newtheorem{assumption}[theorem]{Assumption}
\theoremstyle{remark}

\usepackage{enumitem}

\title{WWW.Serve: Interconnecting Global LLM Services through Decentralization}

\author{Huanyu Wang, Ziyu Xia, Zhuoming Chen, Beidi Chen}
\contact{\{huanyuw, ziyuxia, zhuominc, beidic\}@andrew.cmu.edu}
\affiliation{Carnegie Mellon University}

\abstract{Large language model (LLM) services are mostly centralized, causing inherent scalability bottlenecks and leaving substantial scattered GPU resources underutilized. While decentralized serving could potentially address these limitations, it imposes challenges of \textit{trust}, as the identity and behavior of participants cannot be reliably regulated, and \textit{fairness}, i.e., how to maximize the benefits of all resource providers to improve engagement. However, existing decentralized frameworks predominantly emphasize the rights and protections of users and the cooperative aspect among GPU providers, while \textit{overlooking the inherent competitive dynamics}, imposing substantial constraints on GPU providers, such as requiring them to accept excessive platform-level oversight and to execute all assigned requests with fixed software stacks on fixed hardware configurations. We argue that such assumptions are unrealistic in real-world decentralized environments. To this end, we propose \textbf{WWW.Serve}, a decentralized framework for interconnecting LLM services worldwide. It preserves the flexibility of service providers, allowing them to decide \textit{when, under what policies, and with what resources} they join the decentralized network, while further ensuring their anonymity. In terms of efficiency, WWW.Serve supports self-organizing request dispatch, enabling the network to autonomously allocate requests without centralized coordination. Three key designs are integrated: a blockchain-inspired credit system for trustless collaboration, gossip-driven peer synchronization for flexible participation, and a duel-and-judge mechanism for robust contributor evaluation. Empirically, we show that WWW.Serve incentivizes higher-quality services to obtain greater profit, while improving global SLO (service-level-objective) attainment by up to $\mathbf{1.5\times}$ and lowering latency by $\mathbf{27.6\%}$. Its performance approaches, and in some cases surpasses, centralized scheduling, while fully preserving the benefits of decentralization. These results highlight WWW.Serve as a promising foundation for real-world, decentralized LLM serving.
}
\metadata[Github]{\url{https://github.com/Infini-AI-Lab/WWW.Serve}}
\metadata[Website]{\url{https://infini-ai-lab.github.io/WWW.Serve}}

\begin{document}

\maketitle

\begin{figure}[h]
    \centering
    \begin{subfigure}{0.33\linewidth}
        \centering
        \includegraphics[width=\linewidth]{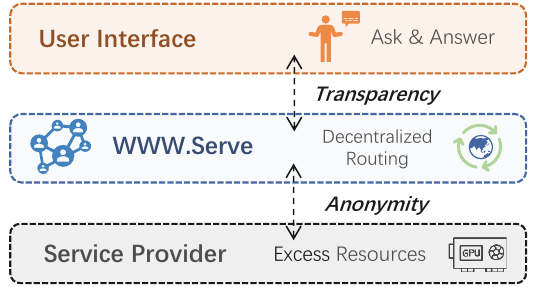}
        \caption{System-level overview.}
        \label{fig:network}
    \end{subfigure}
    \begin{subfigure}{0.66\linewidth}
        \centering
        \includegraphics[width=\linewidth]{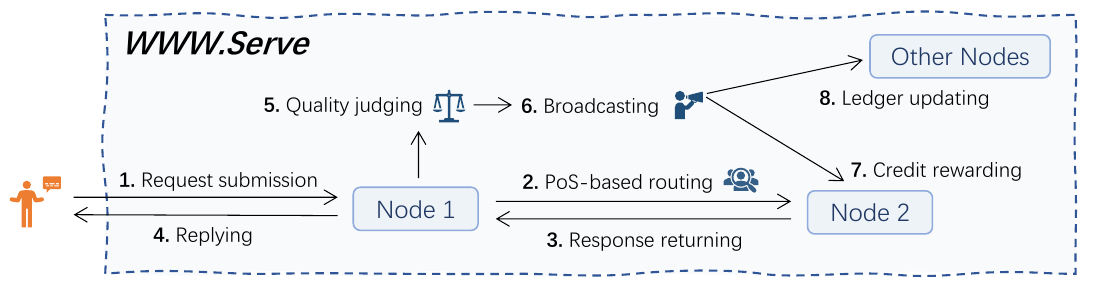}
        \caption{Collaborative request serving workflow.}
        \label{fig:workflow}
    \end{subfigure}
    \caption{WWW.Serve operates as an intermediate decentralized serving layer between users and LLM service providers, offering users access to an open and competitive market of worldwide LLM services while preserving service providers’ anonymity and flexibility. Within WWW.Serve, inference requests follow a collaborative workflow that performs decentralized routing, execution, and quality-aware evaluation.}
\end{figure}

\section{Introduction}

Large language models (LLMs) are becoming popular. With the increasing deployment of LLM services and prices of GPUs, distributed LLM serving has become essential for mitigating workload fluctuations and leveraging potentially idle hardware resources. Centralized scheduling~\citep{sglang,vllm}, however, constrains the engagement of different entities. Therefore, decentralization has long been recognized as an effective paradigm~\citep{survey_dec_learning,survey_dec_training}. By relying on peer-to-peer communication~\citep{think_p2p}, it improves scalability, adapts to dynamic participation, enhances robustness by eliminating single points of failure, and preserves anonymity and privacy~\citep{dec_privacy_preserving,byzantine}.

Despite these apparent advantages, existing decentralized serving systems remain largely impractical in real-world settings: (1) Fundamentally, they predominantly emphasize the rights and protections of users and the cooperative aspect among GPU providers while \textbf{overlooking the inherent competitive dynamics}, namely, that GPU providers, as the holders of the actual computational assets, are naturally incentivized to maximize their own profit. Existing frameworks~\citep{gentorrent} attempt to rely on a small central organization to impose substantial constraints on GPU providers, such as requiring them to accept excessive platform-level oversight~\citep{gentorrent,deserve} and to execute all assigned requests with fixed software stacks~\citep{eccos,petals,helix} on fixed hardware configurations. Although this may theoretically enable better resource allocation, the regulator itself is untrusted, rendering the approach unrealistic in practice. (2) Besides, providers typically maintain their own prioritized workloads and may experience fluctuations in available resources. This highlights the need for flexible, customizable mechanisms that allow providers to determine how they engage with the decentralized system.

Ideally, we desire a decentralized framework that acts like an open, competitive market, allowing providers to decide \textbf{when, under what policies, and with what resources} they join the decentralized network. At the same time, such a framework should: \textbf{1.} provide a well-designed reward mechanism that incentivizes providers to deliver higher-quality services, including faster hardware, more user-oriented scheduling policies, better serving systems, and higher-quality models. Such incentives should further encourage innovation (e.g., in models, systems, or kernels), enabling providers to offer superior services at lower cost. \textbf{2.} enable market-driven exchange of computational capacity, where overloaded nodes can outsource requests while underutilized nodes capitalize on idle resources, allowing compute supply and demand to self-balance through decentralized interactions. \textbf{3.} incorporate a principled routing protocol to improve global efficiency under highly dynamic and unpredictable resource availability. However, to meet these demands, three fundamental questions arise. In the following, we discuss these challenges and outline our key approaches to address them.

\textbf{Question 1.} \textit{How can the system enable trustworthy market-driven trade of computational capacity, i.e., implement reliable request scheduling among anonymous participants without central coordinators?} Achieving this requires a way to quantify each participant’s contributed capacity and use it to guide task allocation. To this end, we introduce a credit-based transaction system that functions as a reputation-like indicator under anonymity: participants earn credits by serving delegated requests and spend them when offloading their own tasks. Request routing is then guided via a Proof-of-Stake-based (PoS) mechanism, in which participants’ staked credits, freely adjust according to their own strategy, determine their likelihood of being selected to execute delegated requests. This design allows high-load servers to offload tasks to relieve pressure and improve user satisfaction, while low-load servers utilize idle resources to earn credits for future offloading. By accumulating credits through contributions, participants effectively engage in a decentralized market for computing power.

\textbf{Question 2.} \textit{How can we incentivize participants to provide high-quality services, thereby improving overall user experience?} In an anonymous network, providers naturally seek to maximize their own gain. These competitive dynamics, however, create the risk that participants may deploy low-quality services to “exploit” the contributions of others, undermining overall system performance. To address this, we must align individual incentives with service quality. To this end, we introduce a duel-and-judge mechanism: a subset of requests is collectively evaluated within the network through pairwise comparison, with the superior response receiving a credit reward and the inferior response incurring a penalty. This design enables dynamic credit redistribution based on service quality. When combined with PoS-based request scheduling, it can be proved that low-quality nodes are gradually phased out of active participation, reinforcing the network’s overall service quality and fostering decentralized incentives for correctness.

\textbf{Question 3.} \textit{How can the system remain robust under highly dynamic and unpredictable resource availability?} In real-world scenarios, individual infrastructures may suffer from hardware failures, network disconnections, or user-driven constraints, all of which lead to unstable participation of resources. To address this challenge, we design a lightweight gossip-driven protocol that enables dynamic online and offline participation. Each participant periodically exchanges availability information with a subset of peers and reconciles discrepancies. Through this protocol, newly joined resources can be quickly integrated into the network, while sudden departures or failures can be rapidly detected. Without relying on central coordinators, lightweight pairwise exchanges allow information updates to diffuse across the network and converge quickly, ensuring stable and reliable service despite the volatility of global-scale resources.

Having addressed these challenges, we introduce \textbf{WWW.Serve}, a decentralized framework for interconnecting global LLM services. In general, our main contributions are:
\begin{itemize}
[itemsep=0.0pt,topsep=0pt,leftmargin=*]
    \item We present \textbf{WWW.Serve}, a fully decentralized system that operates as an open, competitive market of worldwide LLM services, enabling request routing and workload balancing among anonymous LLM servers.
    \item We design three core mechanisms to ensure reliability: a credit-based transaction system for trustless request delegation, a gossip-driven protocol for dynamic peer synchronization, and a duel-and-judge mechanism for contributor evaluation.
    \item We provide a game-theoretic analysis proving that our collaborative framework converges to equilibria that sustain high-quality LLM service even under full anonymity.
    \item Empirical results demonstrate that WWW.Serve achieves near-centralized efficiency, improving global SLO attainment by up to $\mathbf{1.5\times}$ and reducing latency by up to $\mathbf{27.6\%}$, while sustaining robustness under dynamic participation and supporting flexible collaboration policies.
\end{itemize}

The rest of this paper is organized as follows. Section~\ref{sec:related} reviews related work. Sections~\ref{sec:overview} and~\ref{sec:core_design} present the architecture of WWW.Serve and its core mechanisms. Section~\ref{sec:theory} provides a game-theoretic analysis. Section~\ref{sec:application} evaluates the system empirically, with ablation studies presented in Section~\ref{sec:ablation}.

\section{Related Work}
\label{sec:related}

\textbf{Decentralized Computing.} Early volunteer-based platforms~\citep{seti,grid,boinc,foldingathome} demonstrate the feasibility of harnessing distributed resources for large-scale scientific workloads. With the advent of blockchain~\citep{bitcoin}, decentralized frameworks like Ethereum~\citep{ethereum} introduce trustless execution environments where tasks are handled transparently and verifiably through smart contracts. Subsequent systems such as Filecoin~\citep{filecoin} and Golem~\citep{golem} extend this model with incentive mechanisms such as Proof-of-Stake~\citep{ouroboros,casper}, ensuring fair contribution and deterring malicious behavior. These systems highlight the importance of incentive alignment and trustless coordination, motivating our decentralized LLM serving design.

\textbf{Large Language Model Serving.} LLMs demand substantial computational resources, thus are primarily deployed by service providers such as OpenAI~\citep{openai}, Anthropic~\citep{anthropic}, and Microsoft Azure~\citep{azure}, offering users online inference services. Meanwhile, the rapid rise of open-sourced, especially reasoning-oriented models such as DeepSeek-R1~\citep{deepseekr1}, LLaMA~3.1~\citep{llama3}, and Qwen3~\citep{qwen3} series, enables broader community access and deployment, therefore creating massive demand for high-throughput inference services. In response, a spectrum of LLM serving systems has been proposed.

At the single-model level, SGLang~\citep{sglang} and vLLM~\citep{vllm} leverage various advanced techniques to improve request concurrency and maximize inference efficiency. HexGen~\citep{hexgen} and Helix~\citep{helix} provide adaptive scheduling strategies that optimize model deployment and task migration across heterogeneous resources. Furthermore, DistServe~\citep{distserve} partitions prefill and decoding computations across multiple GPUs, while speculative decoding~\citep{speculative1,speculative2,specinfer} and sequence-length-aware scheduling~\citep{proxy_ssjf} offer complementary performance gains. However, these approaches remain inherently centralized and emphasize intra-model performance, without offering systematic solutions for workload balancing across multiple LLM servers.

\textbf{Decentralized LLM Serving.} Recently, decentralized approaches have been further explored, yet they fall short of fully realizing our desired goals. Petals~\citep{petals} supports collaborative deployment of a fixed LLM across volunteer GPUs, limiting flexibility in multi-model scenarios, and cannot adapt to dynamically changing resources. DeServe~\citep{deserve} offers a privacy-preserving offline serving system where users contribute inference capacity collectively, yet still depends on partial centralization for request dispatching and lacks mechanisms to ensure service quality. GenTorrent~\citep{gentorrent} distributes and executes model shards, but relies on trusted organizations to prevent malicious behavior, and therefore does not achieve full decentralization. Other works~\citep{linguachain,llm_blockchain,flock,fedshield} explore secure decentralized training and inference frameworks that integrate cryptographic and blockchain-based mechanisms. While relevant as background, these approaches do not directly address the specific challenges we target.

\section{WWW.Serve's Overview}
\label{sec:overview}

We begin by presenting the overall network architecture of WWW.Serve (Subsection~\ref{subsec:network_architecture}), followed by a description of the request routing process and node design (Subsection~\ref{subsec:request_routing}).

\begin{figure}[t]
    \centering
    \includegraphics[width=0.9\linewidth]{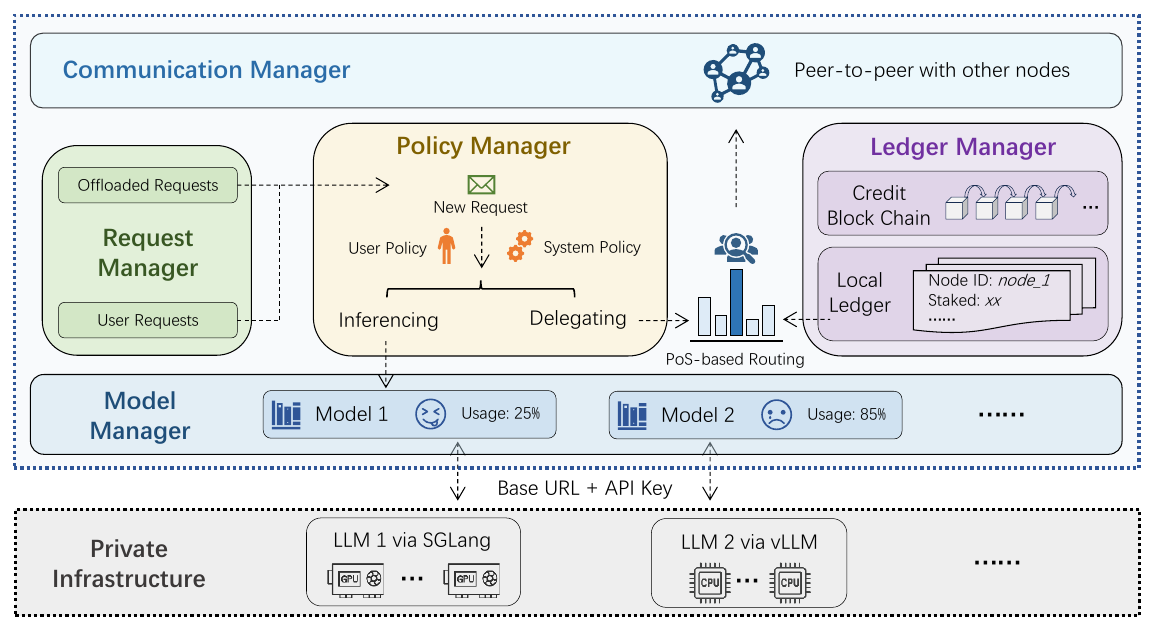}
    \caption{Internal architecture of a single node. Each node within WWW.Serve is organized around five core managers: \emph{Request}, \emph{Policy}, \emph{Ledger}, \emph{Model}, and \emph{Communication}, which together enable PoS-based request routing, policy-driven delegation, and efficient execution over heterogeneous LLM backends.}
    \label{fig:node_design}
\end{figure}

\subsection{Network Architecture}
\label{subsec:network_architecture}

As illustrated in Figure~\ref{fig:network}, WWW.Serve establishes a fully decentralized peer-to-peer network connecting users with LLM service providers.

From the user’s perspective, WWW.Serve provides a seamless serving interface. Users do not need to be aware of the underlying decentralized infrastructure; instead, they simply submit inference requests and wait for responses, just as they would with conventional LLM online services. The framework automatically handles request routing, resource discovery, and response evaluation. This design greatly lowers the barrier to adoption, allowing users to access global LLM services without requiring specialized knowledge of network topology or coordination protocols.

From the service provider’s perspective, WWW.Serve offers a simple yet flexible participation model. Providers can contribute surplus computational resources without exposing sensitive information, while retaining full control and anonymity within the ecosystem. They are free to join or leave at any time, enabling adaptive scheduling and resource allocation. This design encourages broader participation for service providers, converting idle capacity into valuable contributions for LLM serving.

\subsection{Request Routing and Node Design}
\label{subsec:request_routing}

As illustrated in Figure~\ref{fig:workflow} and~\ref{fig:node_design}, WWW.Serve structures inference serving around a decentralized request routing workflow and a modular node architecture.
The routing workflow consists of four key stages:

\textbf{Request admission.} When a user submits an inference request, it first enters the local request queue maintained by the \emph{Request Manager}, which handles both user-originated and delegated requests. This ensures orderly processing while decoupling admission from execution.

\textbf{Scheduling and policy enforcement.} The queued request is then subject to the service provider’s configurable policies. The \emph{Policy Manager} decides whether to execute the request locally or delegate it to other nodes, considering factors such as workload thresholds, willingness to delegate requests, and customized load-balancing rules. This design allows service providers to flexibly participate in collaborative serving while retaining full control over their resources.

\textbf{Executor selection and trust establishment.} If the request is delegated, the node selects a reliable executor. To this end, the \emph{Ledger Manager} provides access to peers’ stake balances. Candidates are sampled via a Proof-of-Stake-based mechanism, where the probability of selection is proportional to their staked credit. Each candidate is further probed to verify its willingness according to its own policy. Once accepted, the request is forwarded, executed locally by the chosen peer, and the response is returned to the originator. The executor is rewarded through a ``credits-for-offloading'' transaction, while the duel-and-judge mechanism further evaluates response quality (details in Subsection~\ref{subsec:credit_system} and Subsection~\ref{subsec:duel_judge}).

\textbf{Backend-agnostic execution.} For locally served requests, the \emph{Model Manager} provides a unified abstraction layer over diverse serving backends. It executes inference, monitors utilization, and preserves intra-model scheduling efficiency. This ensures that heterogeneous resources can be seamlessly integrated into WWW.Serve.

Together, these stages form a request routing pipeline that ensures policy-driven scheduling, trust-aware executor selection, and efficient execution on heterogeneous LLM servers.

\section{Core Mechanisms}
\label{sec:core_design}

In this section, we introduce three core designs of WWW.Serve: (i) the \emph{Credit-based Transaction System} (Subsection~\ref{subsec:credit_system}), which incentivizes and regulates request dispatching; (ii) the \emph{Duel-and-Judge Mechanism} (Subsection~\ref{subsec:duel_judge}), which ensures reliable and trustworthy contributor evaluation; and (iii) the \emph{Policy Framework} (Subsection~\ref{subsec:policy_framework}), which supports flexible policies for collaboration.

\subsection{Credit-based Transaction System}
\label{subsec:credit_system}

Drawing inspiration from real-world transactions, where users pay for premium LLM services (e.g., API token prices), we design a \emph{Credit-based Transaction System} in which each node’s computational resources are represented as transferable credits. These serve as a reputation-like measure that enables dynamic workload exchange while providing economic incentives for active and high-quality participation. Beyond the system itself, credits can be anchored to real-world currency, enabling direct monetization of computational contributions and paving the way for practical deployment of WWW.Serve in commercial large-scale inference services.

\begin{wraptable}{r}{0.36\linewidth}
    \vspace{-1.4em}
    \centering
    \scriptsize
    \renewcommand{\arraystretch}{1.5}
    \caption{Structure of a Credit Block}
    \label{tab:block}
    \begin{tabular}{|l|l|}
        \hline
        \textbf{Field} & \textbf{Description} \\ \hline
        Block ID & Hash of the current block \\ \hline
        Parent ID & Hash of the previous block \\ \hline
        Timestamp & Time of block creation \\ \hline
        Operations & List of credit-related records \\ \hline
        Proposer & Node proposing the block \\ \hline
        Signature & Digital signature \\ \hline
    \end{tabular}
\end{wraptable}

However, traditional transaction mechanisms are not sufficient in decentralized settings. Without a shared, tamper-resistant ledger, nodes can misreport their actions or selectively reveal inconsistent transaction histories to different peers~\citep{bitcoin,blockchainconsensus,consensusage,blockchainreview}. For example, a node might claim the same credits have been spent in multiple transactions (double-spending), or refuse to acknowledge deductions from failed or malicious executions. Since no single entity holds the authoritative record, such inconsistencies can hardly be reconciled, undermining both fairness and trust across the network.

To address this, WWW.Serve adopts a blockchain-inspired ledger. Each node maintains a local \emph{Credit Block Chain} that records activities such as staking and rewarding in tamper-resistant blocks (Table~\ref{tab:block}). Blocks are cryptographically linked, so any modification is immediately detectable. A credit transaction occurs whenever a delegated request is completed. The responsible node records this by creating a new block and broadcasting it to its peers, which independently validate the block. The transaction is finalized once a majority of peers confirm and append it to their local ledgers.

The security of this design relies on two complementary features. First, nodes must stake credits to participate in scheduling, which discourages malicious behavior by putting dishonest nodes’ stakes at risk. Second, decentralized verification ensures that every block is independently validated by multiple peers before being appended to the chain, preventing any single node from manipulating the ledger. Thus, balances are guaranteed to be secure, auditable, and tamper-resistant, all without relying on a centralized authority.

\subsection{Duel-and-Judge Mechanism}
\label{subsec:duel_judge}

\begin{wrapfigure}{r}{0.5\linewidth}
    \vspace{-2em}
    \centering
    \includegraphics[width=0.92\linewidth]{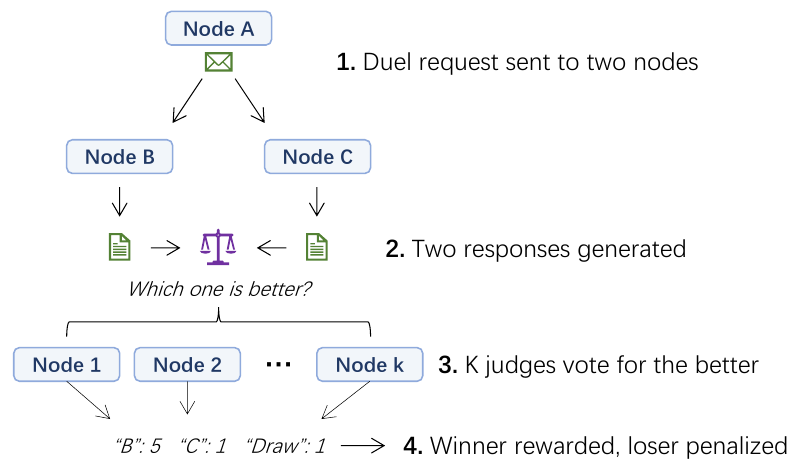}
    \caption{Duel-and-judge mechanism.}
    \label{fig:duel}
\end{wrapfigure}

In our decentralized serving network, participants are anonymous and heterogeneous, with no central authority to verify the quality of their contributions. This raises a fundamental risk: low-quality or even malicious nodes may provide incorrect results, degrading overall service reliability. Prior frameworks~\citep{llmchain,PoQ,gentorrent} rely on verification committees or light evaluation models, but they introduce complexity and privileged roles that limit true decentralization. In response, WWW.Serve introduces a \emph{duel-and-judge mechanism}, enabling peer-driven evaluation of the service quality.

As shown in Figure~\ref{fig:duel}, a small fraction of delegated requests are randomly designated as \emph{duel requests} and dispatched to two executors sampled via our Proof-of-Stake-based selection mechanism. Next, $k$ judges (also selected via PoS) perform pairwise comparisons of the responses. The inferior executor is penalized by losing part of its stake, while the superior executor and the responsible judges earn additional credits. The results of each duel are broadcast and recorded in the credit ledger, ensuring transparency and accountability.

Such a duel-and-judge mechanism offers several key advantages for ensuring reliable and high-quality decentralized serving. First, it leverages a pairwise comparison rather than relying on absolute scores. Prior studies~\citep{mtbench,chatbot-arena,pariksha} demonstrate that pairwise evaluation of LLM outputs yields higher inter-rater agreement and greater robustness, making it a more reliable way to distinguish between competing responses. Second, the involvement of PoS-sampled judge nodes introduces additional impartiality, mitigating risks of collusion and fostering fairness in the evaluation process. Third, the credit redistribution scheme provides strong economic incentives, aligning node behavior with system reliability and thus driving the network toward high-quality operation. A theoretical analysis of the quality evolution is provided in Section~\ref{sec:theory}.

\subsection{Policy Framework}
\label{subsec:policy_framework}

WWW.Serve introduces a policy framework that governs both individual node decisions and collective network behavior, which operates along two complementary dimensions:

\textbf{User-level policies:} enable service providers to manage their resources according to individual objectives. First, each node can freely determine its stake amount, which directly influences its probability of being selected as an executor under the Proof-of-Stake-based scheduling mechanism. This design encourages providers to calibrate their credit commitment according to their willingness and capacity to contribute. Second, nodes may define fine-grained operational conditions for offloading, accepting, or queuing requests at their local backends. For example, one may choose to offload tasks once its local workload surpasses a predefined threshold, to accept external requests only when spare GPU capacity is available, or to prioritize its own user-submitted jobs over delegated ones. Such flexibility not only accommodates heterogeneous resource profiles and business goals, but also fosters a competitive yet cooperative ecosystem where service providers optimize their participation strategies while maintaining overall system efficiency.

\textbf{System-level policies:} serve as global safeguards to preserve fairness and reliability within WWW.Serve, including mechanisms such as Proof-of-Stake-based routing, the credit-based transaction system, gossip-driven peer synchronization, and the duel-and-judge mechanism. These rules provide the necessary trustless foundation, while user-level policies offer flexibility on top of it.

\section{Game-Theory Analysis}
\label{sec:theory}

In this section, we provide a theoremized proof that WWW.Serve converges to a high-quality equilibrium of collaborative LLM services: high-performing nodes accumulate credit over time, whereas low-quality nodes lose exposure and gradually phase out of the system.

\begin{assumption}[Node parameters]
    \label{asm:node}
    For each node $i\in\{1,\dots,N\}$, we have:
    \begin{itemize}
    [itemsep=0.0pt,topsep=0pt]
        \item $q_i \in [0,1]$, the intrinsic probability that node $i$ produces a high-quality response;
        \item $c_i>0$, the per-request operational cost of node $i$;
        \item $s_i(t) \geq 0$, the stake of node $i$ at time $t$.
    \end{itemize}
\end{assumption}

\begin{assumption}[System parameters]
    \label{asm:system}
    The system-level constants are:
    \begin{itemize}
    [itemsep=0.0pt,topsep=0pt]
        \item $\lambda>0$, the delegated request arrival rate;
        \item $R>0$, the guaranteed base reward per delegated request;
        \item $p_d \in [0,1]$, the probability that a delegated request is selected as a duel;
        \item $R_{add}>0$, the additional reward for winning a duel;
        \item $P>0$, the penalty for losing a duel.
    \end{itemize}
\end{assumption}

\begin{assumption}[PoS-based selection and duel-and-judge mechanism]
    \label{asm:duel}
    We write the PoS selection probability of node $i$ and the selection-weighted global average quality as
    \[
    p_i(t) \;=\; \frac{s_i(t)}{\sum_{j=1}^N s_j(t)}, 
    \qquad
    \overline{Q}(t) \;=\; \sum_{i=1}^N p_i(t) \, q_i.
    \]
    To capture the intuition that a higher network average quality $\overline{Q}(t)$ makes it harder for any individual node to stand out, we model the probability that node $i$ wins the duel as
    \[
    Q_i(t) \;=\; \tfrac{1}{2} \big(1 + q_i - \overline{Q}(t)\big) \in [0, 1].
    \]
\end{assumption}

\begin{assumption}[Stake adjustment]
    \label{asm:stake}
    Rational participants adjust their stakes proportionally to realized expected payoffs. Concretely, for some growth constant $\eta>0$ we assume
    \[
    \dot s_i(t) \;=\; \eta \, \pi_i(t),
    \]
    where $\pi_i(t)$ denotes node $i$'s expected payoff rate (defined below in Lemma~\ref{lemma:payoff}).
\end{assumption}

\begin{lemma}[Expected node payoff]
    \label{lemma:payoff}
    Under Assumptions~\ref{asm:node}--\ref{asm:duel}, the expected payoff of node $i$ from serving a single delegated request is
    \[
    \Delta_i(t) \;=\; (R - c_i) + p_d \big[Q_i(t) \, R_{add} - (1-Q_i(t)) \, P\big].
    \]
    Consequently, the expected payoff rate of node $i$ under delegated request arrival rate $\lambda$ and PoS selection probability $p_i(t)$ is
    \[
    \pi_i(t) \;=\; \lambda \, p_i(t) \, \Delta_i(t).
    \]
\end{lemma}

\begin{proof}
    A single delegated request always yields the base reward $R$ and incurs cost $c_i$, hence the guaranteed net term $(R-c_i)$. With probability $p_d$ the request becomes a duel; conditional on a duel, the expected duel outcome for node $i$ equals $Q_i(t)\,R_{add}-(1-Q_i(t))\,P$. Adding these terms gives $\Delta_i(t)$. Multiplying by the delegated request arrival rate $\lambda$ and the selection probability $p_i(t)$ yields the stated expression for $\pi_i(t)$.
\end{proof}

\begin{proposition}[Single-node stake-share dynamics]
    \label{prop:single}
    Under Assumptions~\ref{asm:node}--\ref{asm:stake}, the stake share of node $i$ evolves according to
    \begin{equation}\label{eq:single_payoff_prop}
    \dot p_i(t) \;=\; \frac{\eta \, \lambda}{S(t)} \, p_i(t) \big(\Delta_i(t) - \overline{\Delta}(t) \big),
    \end{equation}
    where $S(t)=\sum_{j} s_j(t)$ is the total stake in the network, and $\overline{\Delta}(t) = \sum_j p_j(t) \Delta_j(t)$ represents the overall average expected payoff.
\end{proposition}

\begin{proof}
    Differentiate $p_i(t)=s_i(t)/S(t)$ to obtain
    \[
    \dot p_i(t) \;=\; \frac{\dot s_i(t) S(t) - s_i(t)\dot S(t)}{S(t)^2}.
    \]
    By Assumption~\ref{asm:stake} we have $\dot s_i(t)=\eta\pi_i(t)=\eta\lambda p_i(t)\Delta_i(t)$, and summing over $i$ yields
    \[
    \dot S(t)=\sum_{j}\dot s_j(t)=\eta\lambda\sum_j p_j(t)\Delta_j(t)=\eta\lambda\,\overline{\Delta}(t).
    \]
    Substituting these into the derivative and simplifying gives \eqref{eq:single_payoff_prop}.
\end{proof}

\begin{proposition}[Group-level stake-share dynamics]
    \label{prop:group}
    Let $\mathbb{H}\subseteq\{1,\dots,N\}$ be any subset of nodes, and define its group-level stake share as
    \[
    p_H(t) \;=\; \sum_{i\in H} p_i(t).
    \]
    Define the within-group and outside-group average payoffs as
    \[
    \overline{\Delta}_H(t) \;=\; \frac{1}{p_H(t)} \sum_{i\in H} p_i(t) \, \Delta_i(t),
    \qquad
    \overline{\Delta}_{\neg H}(t) \;=\; \frac{1}{1-p_H(t)} \sum_{j\notin H} p_j(t) \, \Delta_j(t).
    \]
    Then the group-level stake share evolves according to
    \begin{equation}\label{eq:group_payoff_prop}
    \dot p_H(t) \;=\; \frac{\eta \, \lambda}{S(t)} \, p_H(t) (1-p_H(t)) \big(\overline{\Delta}_H(t) - \overline{\Delta}_{\neg H}(t)\big).
    \end{equation}
\end{proposition}

\begin{proof}
    Summing \eqref{eq:single_payoff_prop} over $i\in H$ yields
    \[
    \dot p_H(t) \;=\; \frac{\eta\lambda}{S(t)}\Big(\sum_{i\in H} p_i(t)\Delta_i(t)-p_H(t)\overline{\Delta}(t)\Big).
    \]
    Write the network average $\overline{\Delta}(t)$ as the convex combination
    \[
    \overline{\Delta}(t) \;=\; p_H(t)\overline{\Delta}_H(t)+(1-p_H(t))\overline{\Delta}_{\neg H}(t).
    \]
    Substituting this into the previous display and simplifying produces \eqref{eq:group_payoff_prop}.
\end{proof}

\begin{theorem}[High-quality equilibrium]
    \label{thm:equilibrium}
    Under Assumptions~\ref{asm:node}--\ref{asm:stake}, the network converges to a high-quality equilibrium, driven by a subset of superior nodes, thereby promoting reliable and high-quality LLM services.
\end{theorem}

\begin{proof}
    From Proposition~\ref{prop:group}, if there exists a subset $\mathbb{H}$ and a time $T$ such that for all $t\ge T$,
    \[
    \overline{\Delta}_H(t)>\overline{\Delta}_{\neg H}(t),
    \]
    then $\dot p_H(t)>0$, hence $p_H(t)$ is strictly increasing for $t\ge T$. Consequently, high-quality nodes progressively accumulate credit while low-quality nodes lose influence, creating incentives for participants to provide superior services and guiding the network toward reliable and high-quality LLM serving.
\end{proof}

\section{Empirical Evaluation}
\label{sec:application}

In this section, we evaluate WWW.Serve under diverse configurations and workload scenarios (implementation details are provided in Appendix~\ref{app:implementation}):
\begin{itemize}
[itemsep=0.0pt,topsep=0pt,leftmargin=*]
    \item In Subsection~\ref{subsec:efficiency}, we show that WWW.Serve improves global SLO attainment by up to $\mathbf{1.5\times}$ and reduces latency by $\mathbf{27.6\%}$ compared to single-node deployment, achieving efficiency close to centralized scheduling.
    \item In Subsection~\ref{subsec:dynamic}, we demonstrate that WWW.Serve handles dynamic participation gracefully, maintaining service continuity as computational resources join or leave.
    \item In Subsection~\ref{subsec:quality_incentivization}, we confirm that WWW.Serve effectively incentivizes superior LLM services, with credit accumulation favoring higher-quality models, more advanced serving systems, and faster hardware.
\end{itemize}

\subsection{Scheduling Efficiency}
\label{subsec:efficiency}

\begin{figure}[t]
    \centering
    \includegraphics[width=0.5\linewidth]{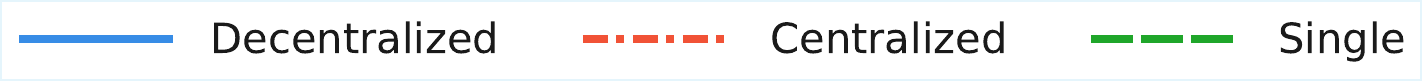}\\
    \includegraphics[width=0.24\linewidth]{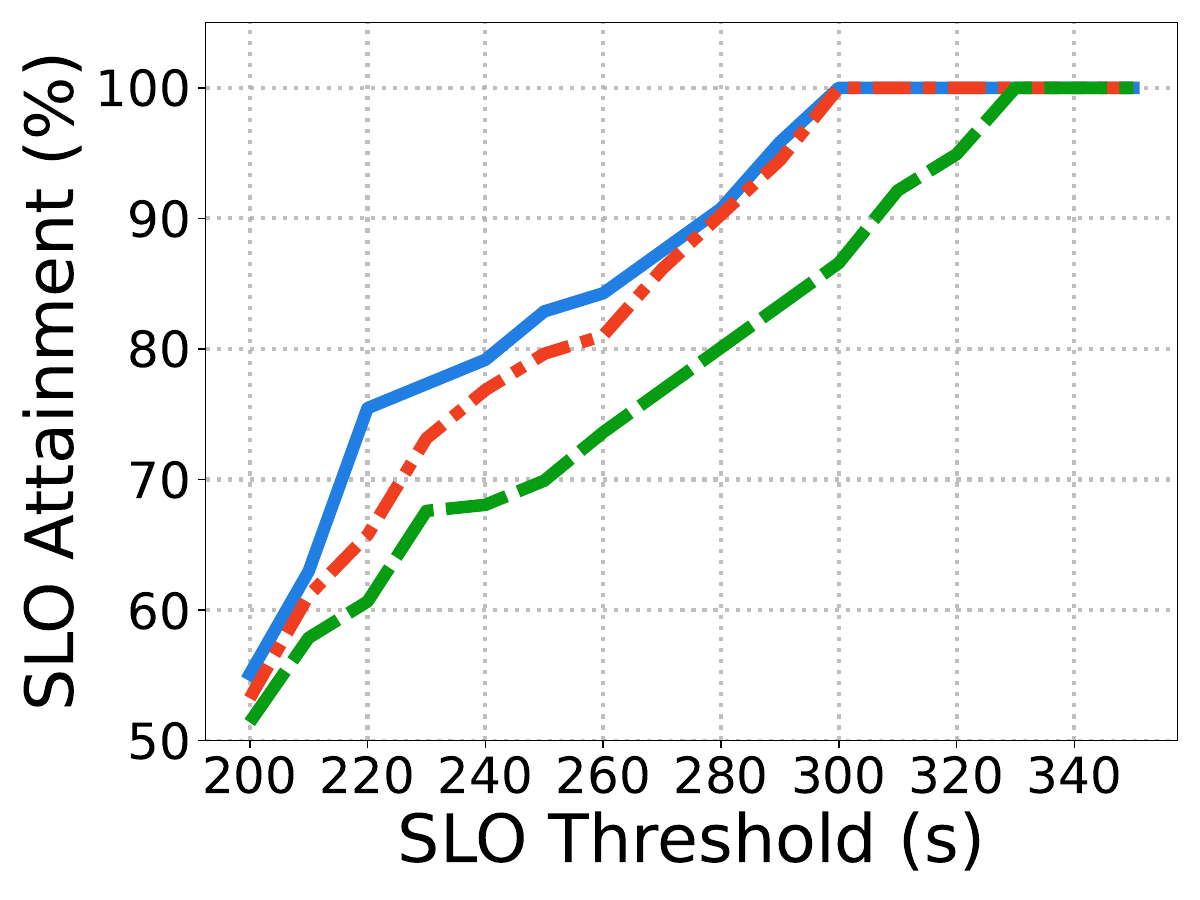}
    \includegraphics[width=0.24\linewidth]{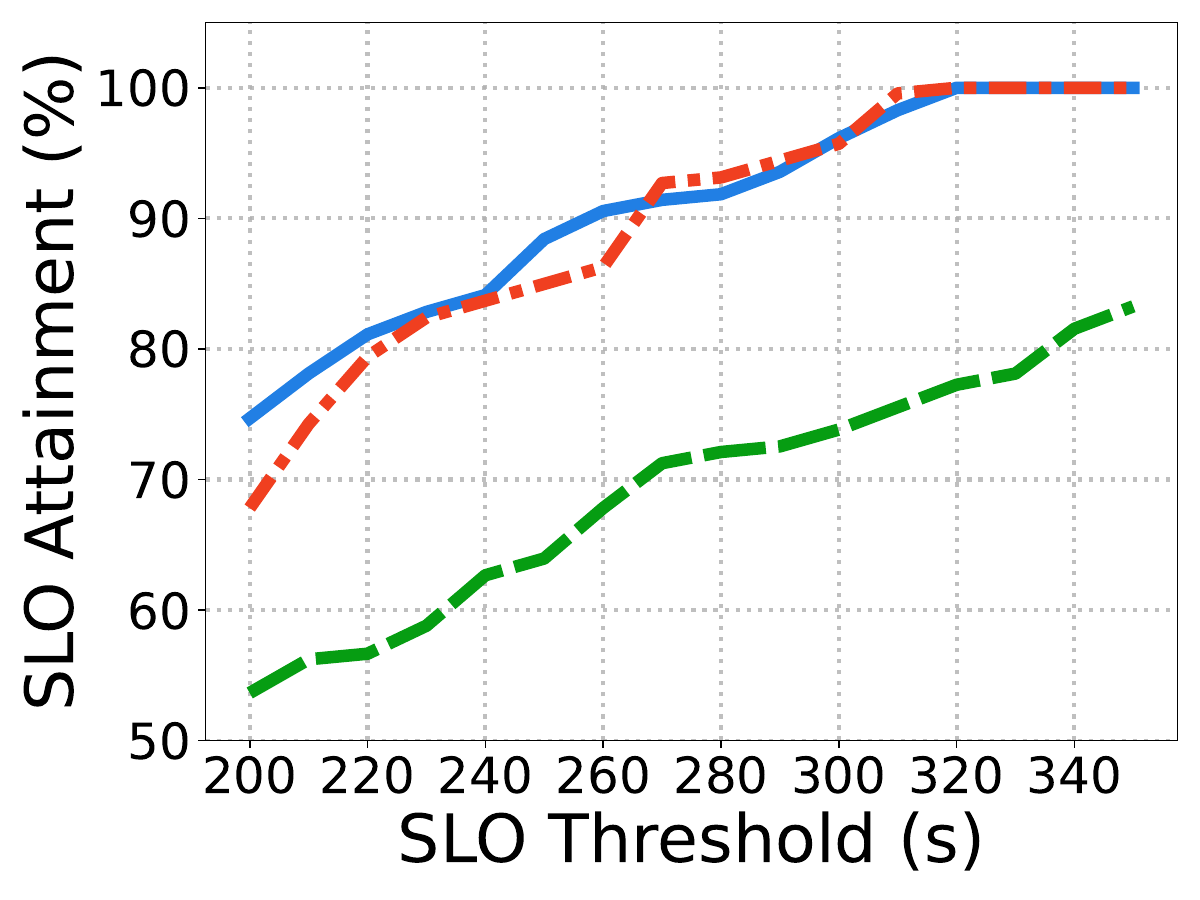}
    \includegraphics[width=0.24\linewidth]{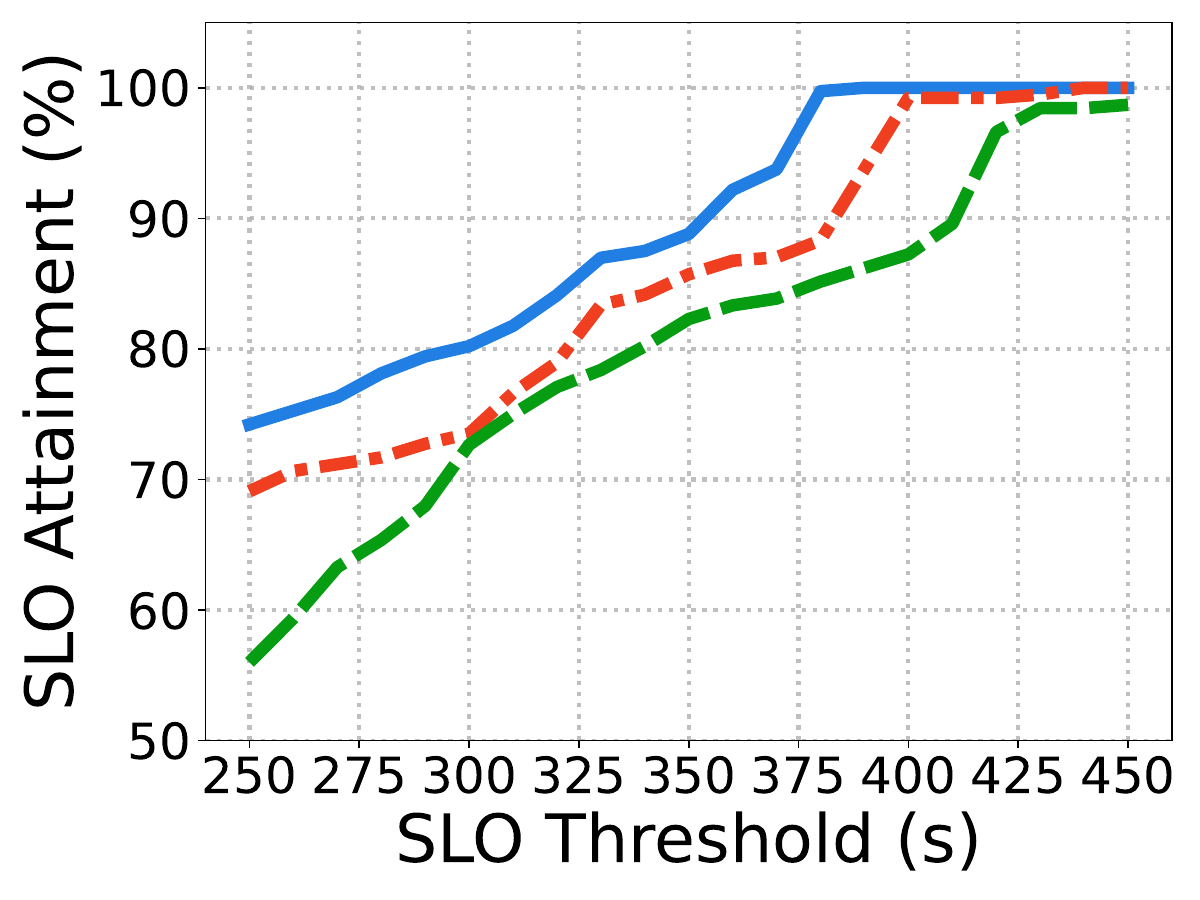}
    \includegraphics[width=0.24\linewidth]{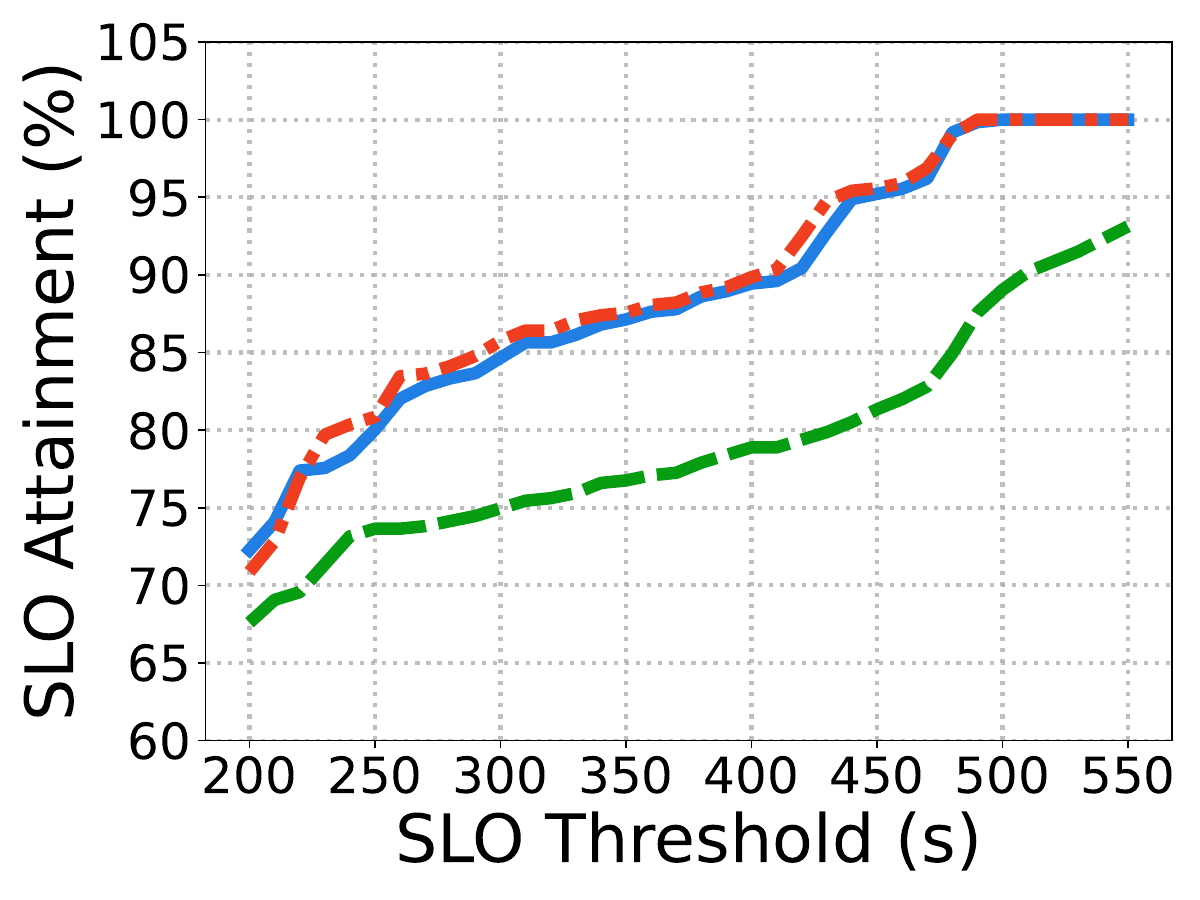}
    \caption{Comparison of global SLO attainment across single-node, centralized, and decentralized (WWW.Serve) deployments under four different experimental settings (Settings 1-4 from left to right; see Appendix~\ref{app:setting} for details).}
    \label{fig:basic}
\end{figure}

We first designed a variety of deployment scenarios (details in Appendix~\ref{app:setting}), covering heterogeneous models, diverse GPU hardware, and multiple serving backends. Each node experienced alternating peak and idle periods, simulating realistic fluctuations in service demand. We compared three deployment strategies: single, centralized, and our decentralized scheduling, and measured global Service Level Objective (SLO) attainment (i.e., the proportion of requests completed within predefined latency thresholds) along with the average request latency.

\begin{wraptable}{r}{0.43\linewidth}
    \vspace{-1.5em}
    \centering
    \scriptsize
    \setlength{\tabcolsep}{2.2mm}
    \renewcommand{\arraystretch}{1.2}
    \caption{Average request latency comparing different scheduling strategies.}
    \label{tab:latency}
    \begin{tabular}{c|ccc}
        \toprule
        \multirow{2}{*}{\textbf{Setting}} & \multicolumn{3}{c}{\textbf{Avg. Latency (s)}} \\
        & Single & Centralized & Decentralized \\
        \midrule
        Setting 1 & 200.380 & 188.419 & \textbf{184.400} \\
        Setting 2 & 226.578 & \textbf{168.221} & 168.485 \\
        Setting 3 & 237.925 & 206.123 & \textbf{198.306} \\
        Setting 4 & 241.042 & \textbf{169.896} & 174.592 \\
        \bottomrule
    \end{tabular}
\end{wraptable}

As shown in Figure~\ref{fig:basic}, across all experimental settings, WWW.Serve consistently outperforms single-node deployment and closely matches, in some cases even surpasses, centralized scheduling in terms of SLO attainment. Table~\ref{tab:latency} further demonstrates that this efficiency translates into substantially lower request latency. Together, these results highlight a key advantage of WWW.Serve: it achieves near-centralized scheduling efficiency without compromising the privacy and autonomy afforded by decentralization.

\subsection{Dynamic Participation}
\label{subsec:dynamic}

\begin{figure}[t]
    \centering
    \includegraphics[width=0.7\linewidth]{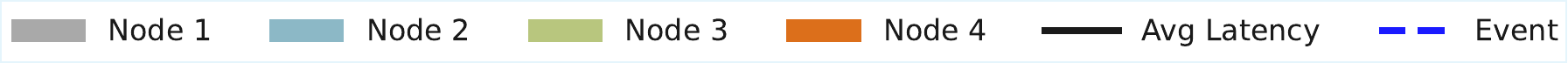}\\
    \begin{subfigure}{0.45\linewidth}
        \centering
        \includegraphics[width=0.9\linewidth]{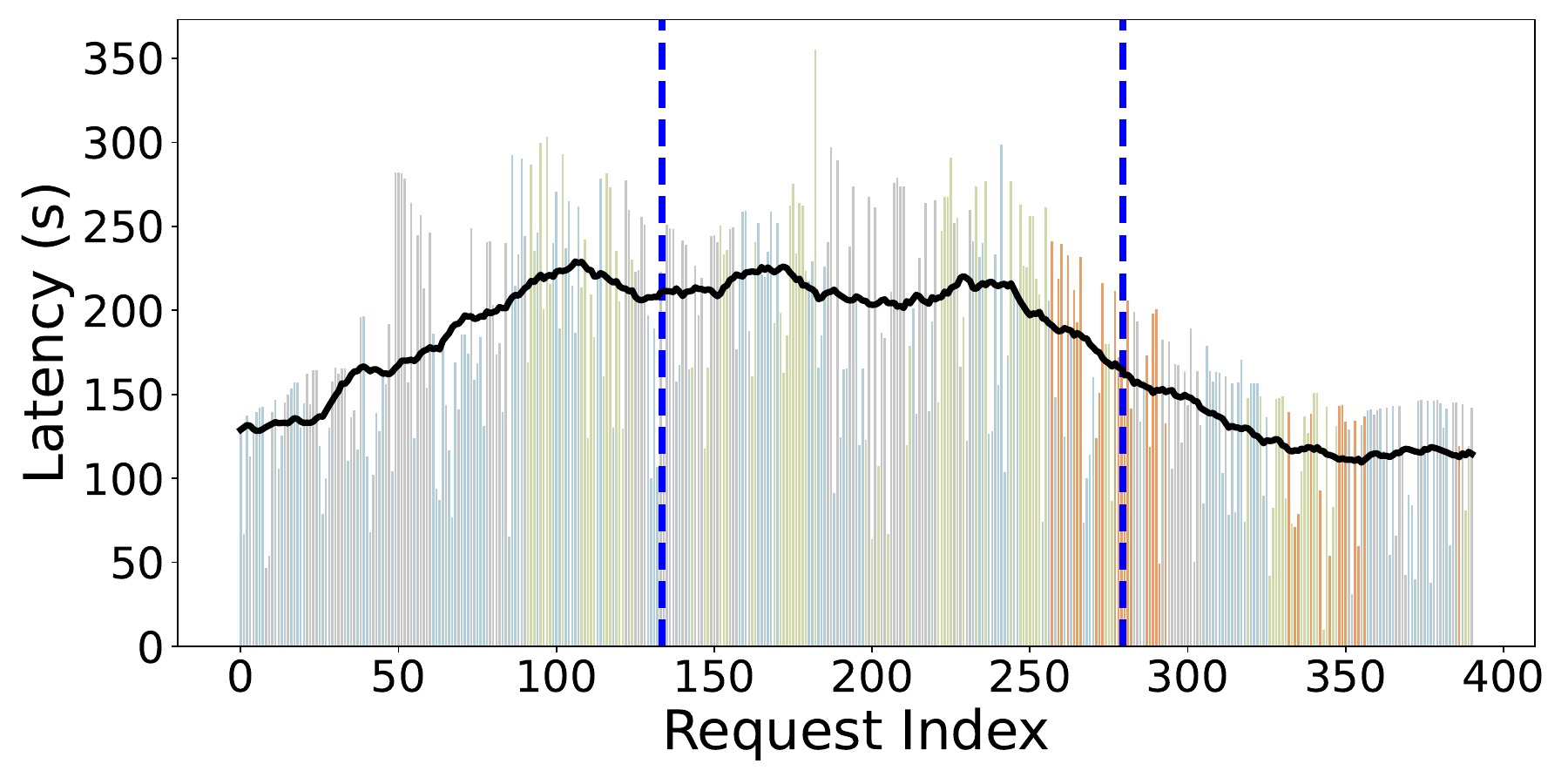}
        \caption{Node Join}
        \label{fig:node-join}
    \end{subfigure}
    \begin{subfigure}{0.45\linewidth}
        \centering
        \includegraphics[width=0.9\linewidth]{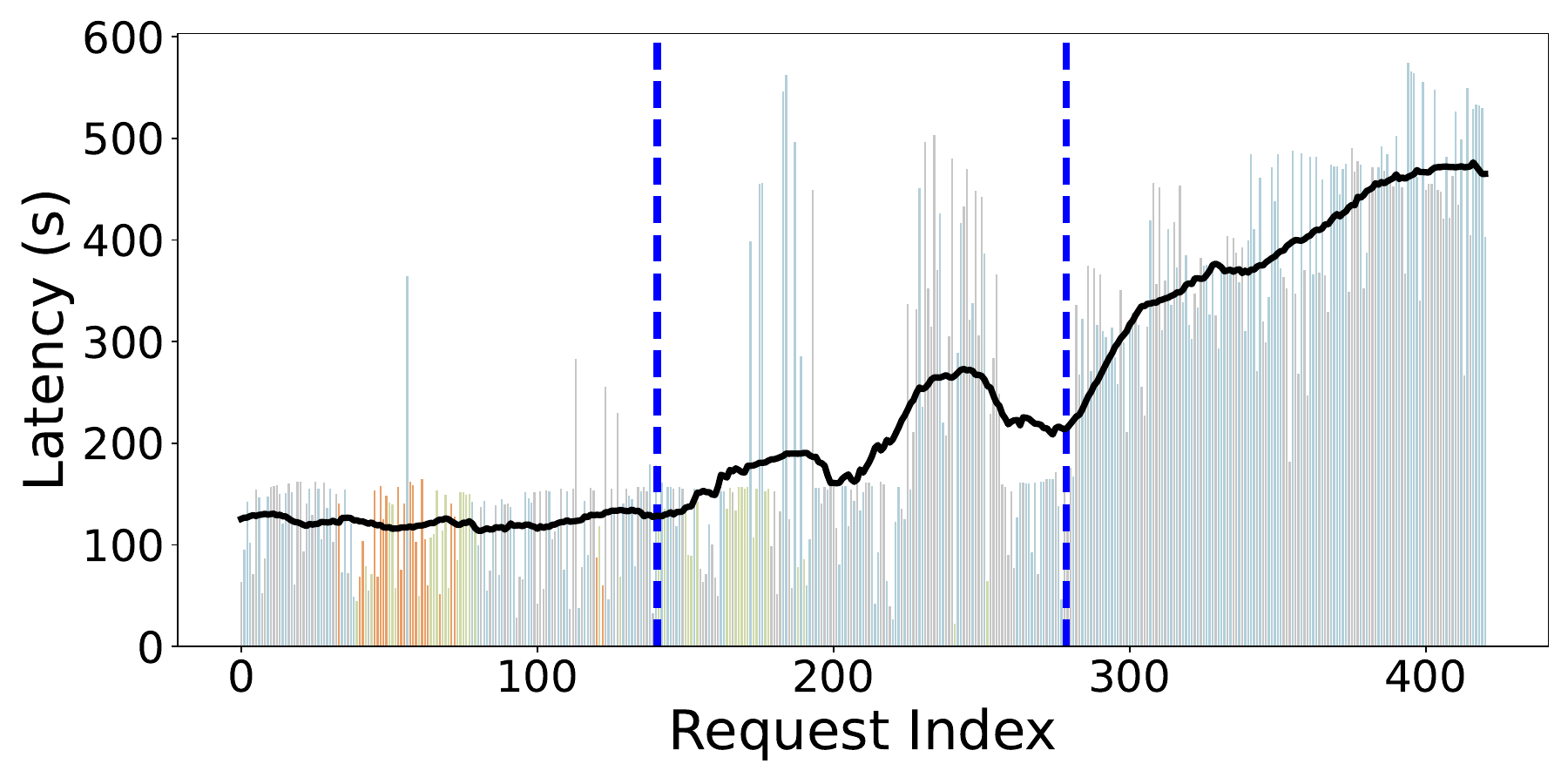}
        \caption{Node Leave}
        \label{fig:node-leave}
    \end{subfigure}
    \caption{Request latency under dynamic participation. Blue lines indicate node join/leave events; black lines show the windowed average latency.}
\end{figure}

WWW.Serve is designed to operate under highly dynamic and unpredictable resource availability in real-world scenarios. Thus, we evaluate its ability to adapt to arbitrary node arrivals and departures.

Figure~\ref{fig:node-join} illustrates nodes joining the network sequentially, starting with two active nodes. When the workload temporarily exceeds available resources, request latencies initially rise. As new nodes are integrated, the gossip-based protocol quickly detects them and redistributes requests, leading to a clear reduction in latency. Conversely, Figure~\ref{fig:node-leave} starts with four nodes and two leave the network sequentially. As the average load increases, the remaining nodes become increasingly saturated, resulting in a sharp rise in overall latency. These results demonstrate that WWW.Serve can dynamically adapt its workload distribution to both node arrivals and departures without a central coordinator, ensuring service continuity in unstable environments.

\subsection{Quality Incentivization}
\label{subsec:quality_incentivization}

\begin{figure}[t]
    \centering
    \begin{subfigure}{0.49\linewidth}
        \centering
        \includegraphics[width=\linewidth]{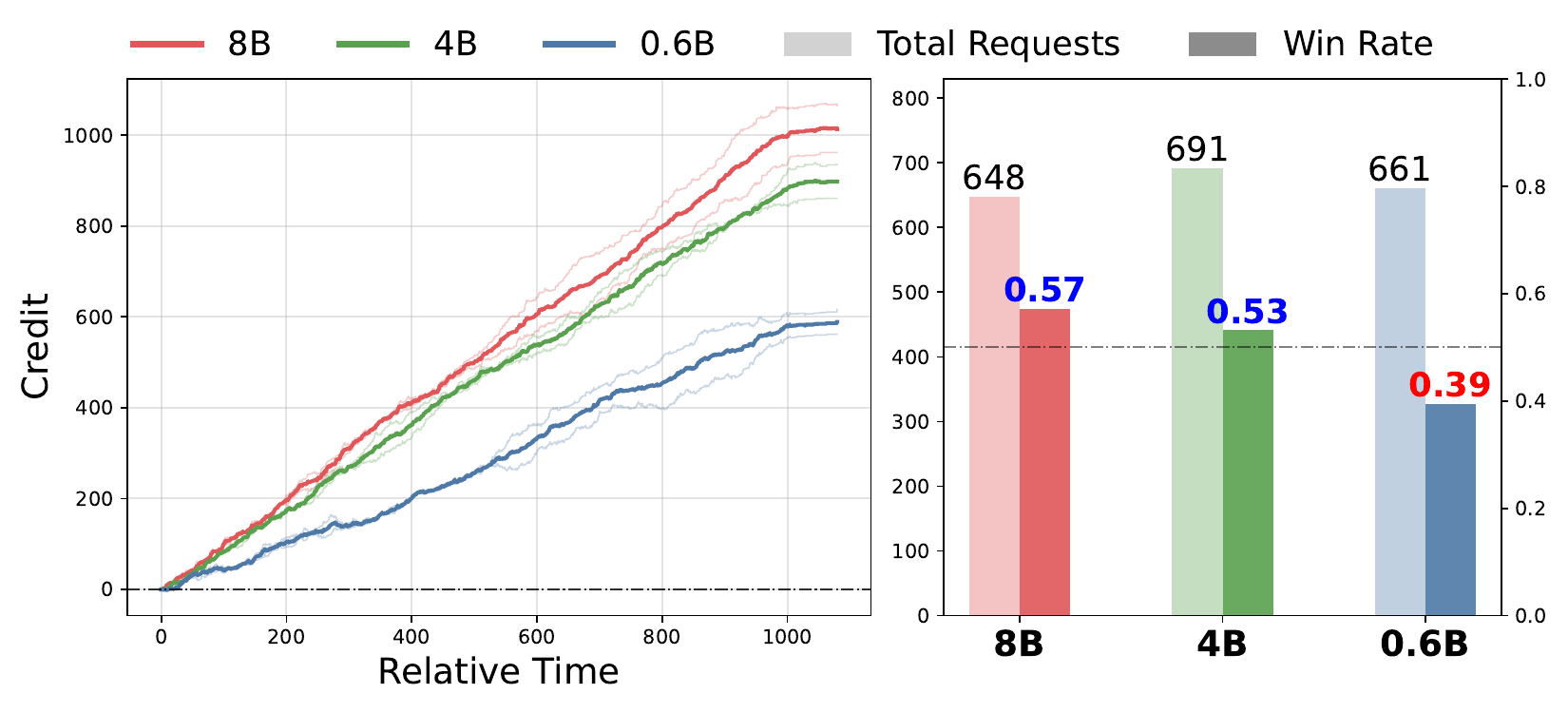}
        \caption{Model capacity}
        \label{fig:quality-model}
    \end{subfigure}
    \hfill
    \begin{subfigure}{0.49\linewidth}
        \centering
        \includegraphics[width=\linewidth]{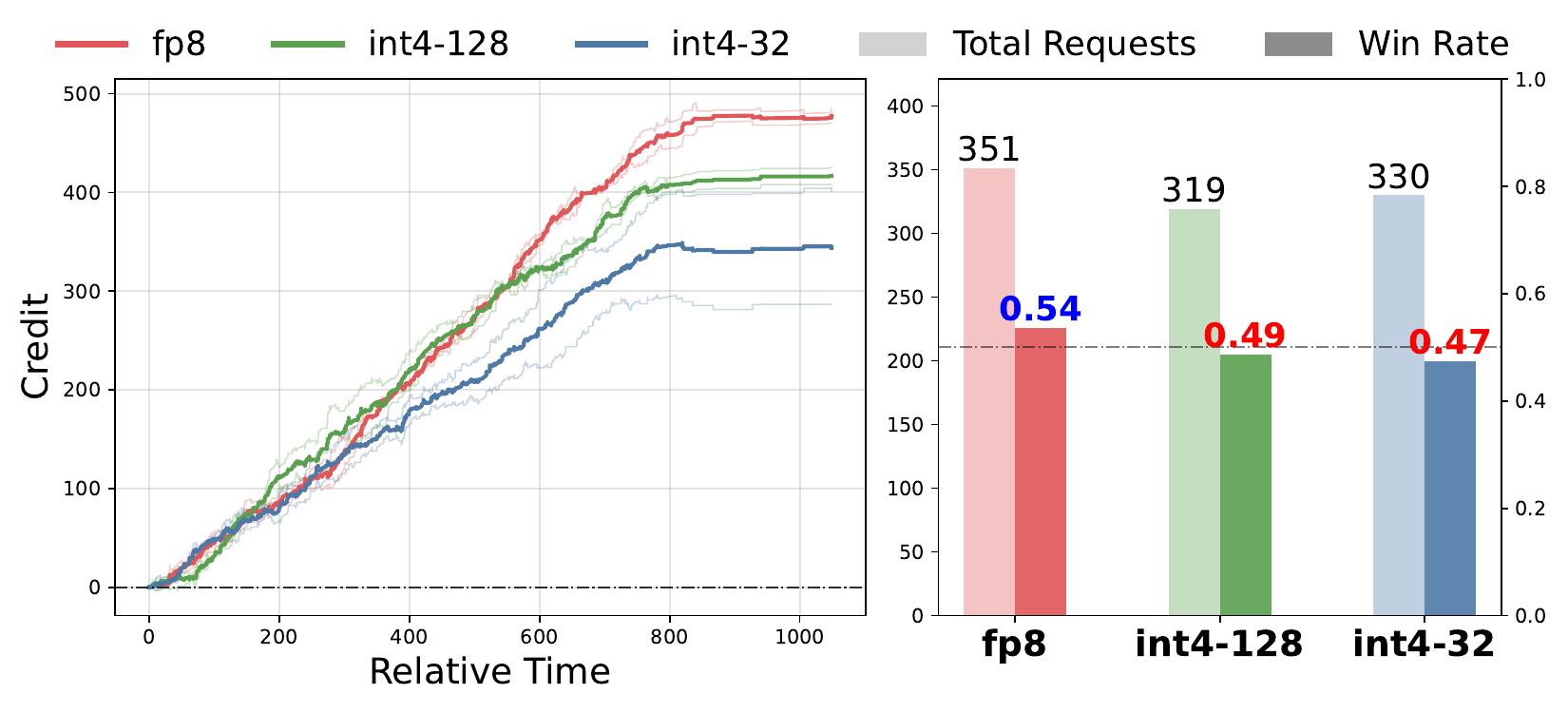}
        \caption{Quantization}
        \label{fig:quality-quant}
    \end{subfigure}\\
    \begin{subfigure}{0.49\linewidth}
        \centering
        \includegraphics[width=\linewidth]{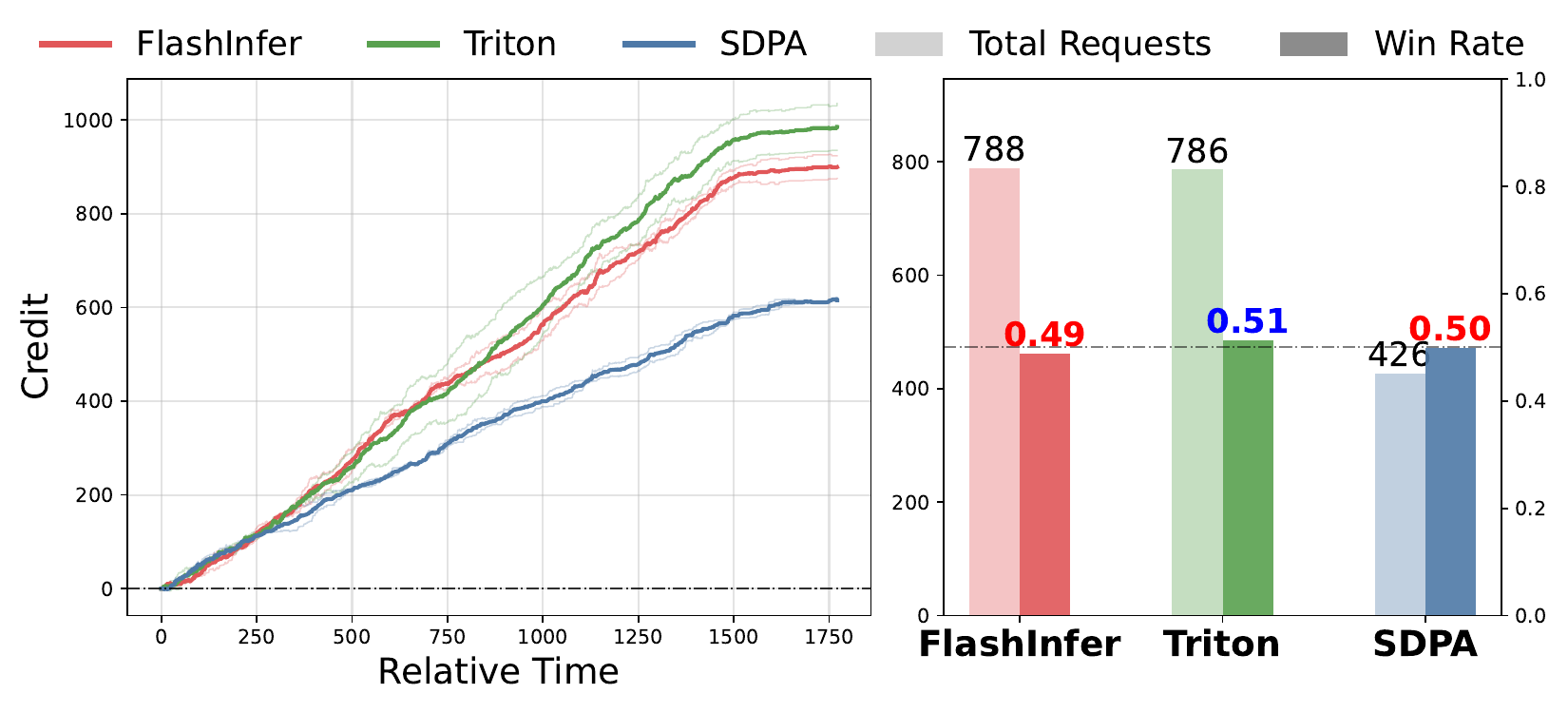}
        \caption{Serving efficiency}
        \label{fig:quality-attn}
    \end{subfigure}
    \hfill
    \begin{subfigure}{0.49\linewidth}
        \centering
        \includegraphics[width=\linewidth]{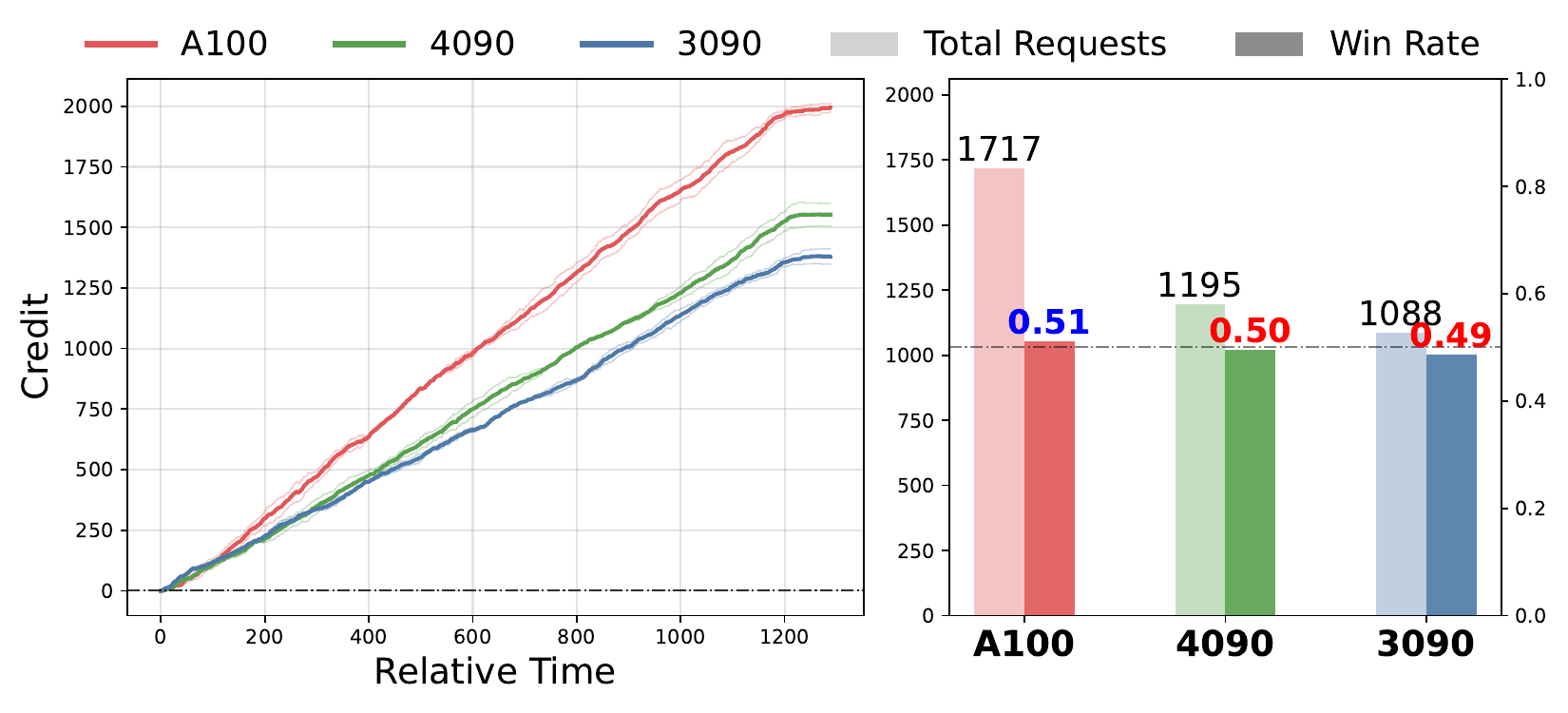}
        \caption{Hardware}
        \label{fig:quality-hw}
    \end{subfigure}
    \caption{Credit dynamics under heterogeneous node capabilities. For each setting, the left panel shows the evolution of credit over time, while the right panel reports the total number of served requests and the duel-and-judge win rate.}
\end{figure}

WWW.Serve is designed to incentivize providers to deliver high-quality LLM services. To empirically validate this objective, we conduct four controlled experiments. In each experiment, we deploy three classes of nodes with heterogeneous capabilities, with two replicas per class to mitigate instance-level randomness. These experiments examine how the credit mechanism rewards: (i) higher-quality models, (ii) more advanced serving systems, and (iii) faster hardware.

\textbf{Model capacity.} (Figure~\ref{fig:quality-model}) Nodes serve Qwen3 models of different sizes (8B, 4B, and 0.6B). Larger models consistently achieve higher duel-and-judge win rates (0.57, 0.53, and 0.39), indicating superior response quality. This advantage is reflected in the credit trajectories, with credit accumulating progressively faster as model size increases from 0.6B to 8B.

\textbf{Quantization.} (Figure~\ref{fig:quality-quant}) Nodes serve the same Qwen3-8B model under different quantization strategies based on TorchAO~\citep{torchao}: fp8wo, int4wo-128, and int4wo-32. More aggressive quantization degrades response quality, resulting in lower duel-and-judge win rates (0.54, 0.49, and 0.47). Consequently, nodes with heavier quantization exhibit slower credit accumulation.

\textbf{Serving efficiency.} (Figure~\ref{fig:quality-attn}) Nodes serve the same Qwen3-8B model using different attention backends: FlashInfer~\citep{flashinfer}, Triton~\citep{triton}, and SDPA~\citep{sdpa}. More efficient backends achieve substantially higher request throughput and serve more requests (788, 786, and 426). Since duel-and-judge win rates remain comparable across backends, differences in credit accumulation are primarily driven by throughput advantages.

\textbf{Hardware resources.} (Figure~\ref{fig:quality-hw}) Nodes serve the same Qwen3-8B model on different GPUs: A100, RTX4090, and RTX3090. Nodes equipped with higher computational capacity and larger GPU memory achieve higher request throughput and serve more concurrent requests (1717, 1195, and 1088), which in turn leads to faster credit accumulation.

Across all four experiments, WWW.Serve consistently aligns credit accumulation with both service quality and serving efficiency: nodes with higher-quality models or less aggressive quantization achieve higher duel-and-judge win rates, while nodes with more efficient serving systems or stronger hardware serve more requests. Together, these factors drive credit accumulation, showing that WWW.Serve effectively incentivizes better LLM service.

\section{Ablation Study}
\label{sec:ablation}

In this section, we first analyze the overhead of the duel-and-judge mechanism (subsection~\ref{subsec:duel_overhead}), and then study the effects of user-level policy configurations on system behavior (subsec~\ref{subsec:policies}).

\subsection{Overhead of Duel-and-Judge Mechanism}
\label{subsec:duel_overhead}

To assess the overhead introduced by the duel-and-judge mechanism, we provide a theoretical analysis, followed by an empirical evaluation of latency and SLO attainment under different duel rates.

We first quantify the incremental request load. Let:
\begin{itemize}
[itemsep=0.0pt,topsep=0pt]
    \item $N$: total number of user requests across all nodes;
    \item $\alpha$: request delegation rate ($\alpha N$ requests are offloaded for remote inference);
    \item $p_d$: duel rate (a fraction $p_d$ of delegated requests are selected as duel requests);
    \item $k$: number of judges per duel.
\end{itemize}
Each duel request triggers one challenger inference and $k$ judge evaluations, contributing $(1+k)$ additional requests. Thus, the expected number of extra requests introduced by the duel-and-judge mechanism is
\[
N \alpha \, p_d \, (1+k),
\]
which remains modest compared to the overall serving workload.

\begin{wrapfigure}{r}{0.48\linewidth}
    \vspace{-1.6em}
    \centering
    \includegraphics[width=0.7\linewidth]{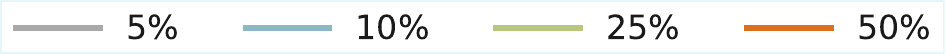}\\
    \includegraphics[width=0.48\linewidth]{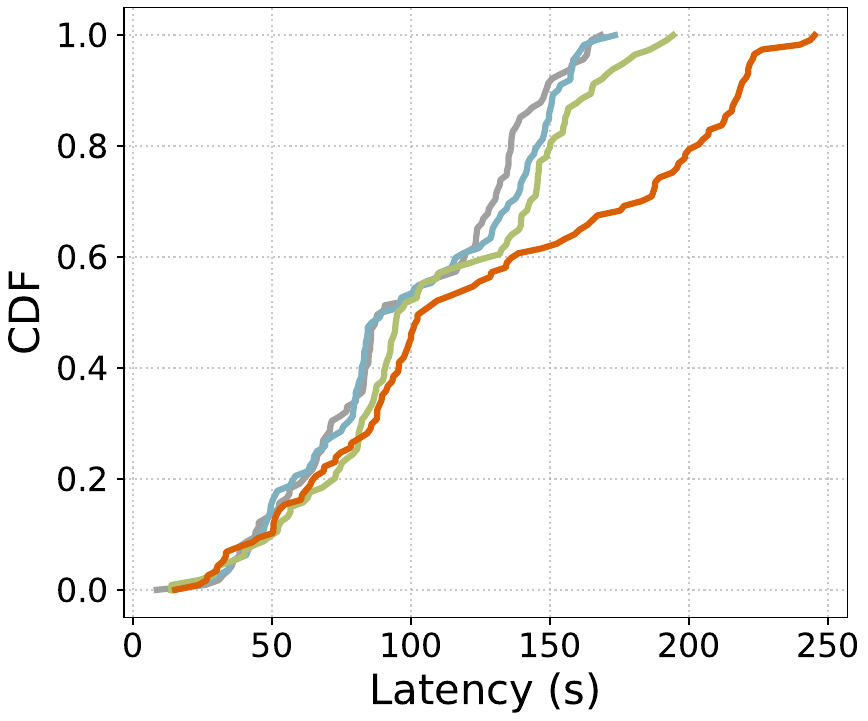}
    \includegraphics[width=0.48\linewidth]{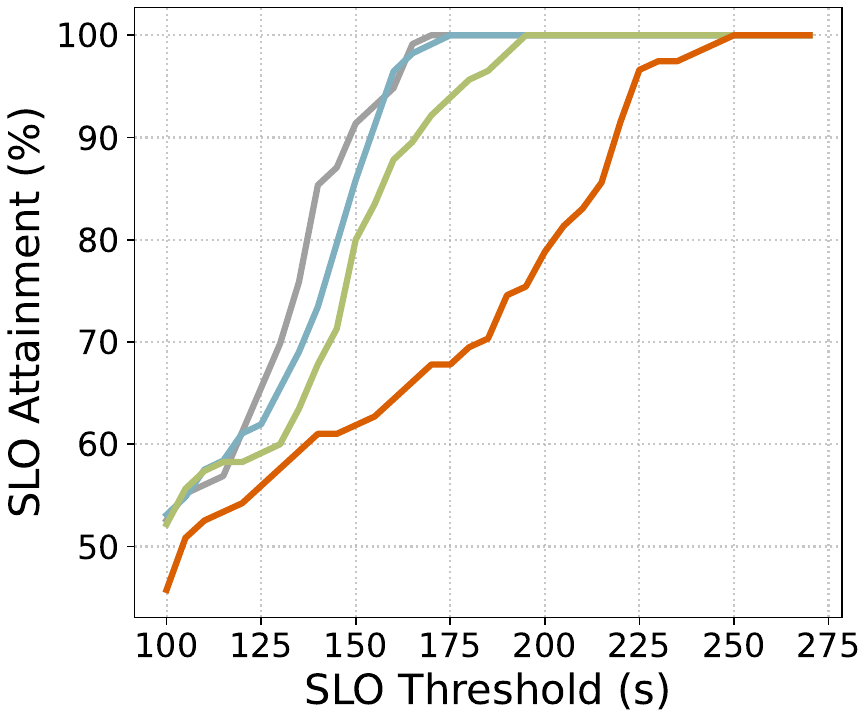}
    \caption{Latency CDF (left) and SLO attainment (right) for different duel rates.}
    \label{fig:duel_latency_slo}
\end{wrapfigure}

To empirically evaluate the effect of duel rate on system performance, we conduct an ablation study using four nodes, with $k=2$ judges per duel. Requests are uniformly issued by a dedicated requester-only node. This configuration intentionally imposes a higher load than typical deployments: fewer nodes yet multiple judges per duel amplify the relative overhead. As shown in Figure~\ref{fig:duel_latency_slo}, duel probabilities of 5\%, 10\%, and 25\% yield nearly identical latency CDFs and SLO attainment curves, indicating that moderate duel rates introduce minimal overhead.

\subsection{Impact of User-Level Policies}
\label{subsec:policies}

\begin{figure}[t]
    \centering
    \begin{subfigure}{0.67\linewidth}
        \centering
        \includegraphics[width=0.6\linewidth]{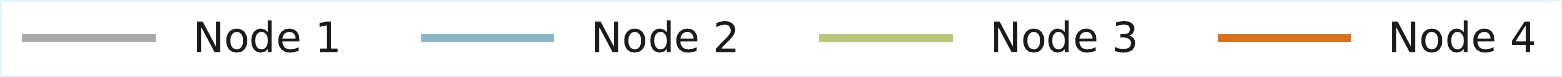}
    \end{subfigure}
    \hfill
    \begin{subfigure}{0.32\linewidth}
        \centering
        \includegraphics[width=\linewidth]{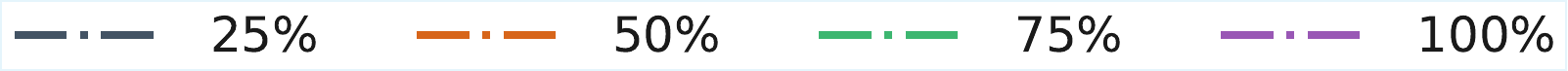}
    \end{subfigure}\\
    \begin{subfigure}{0.33\linewidth}
        \centering
        \includegraphics[width=\linewidth]{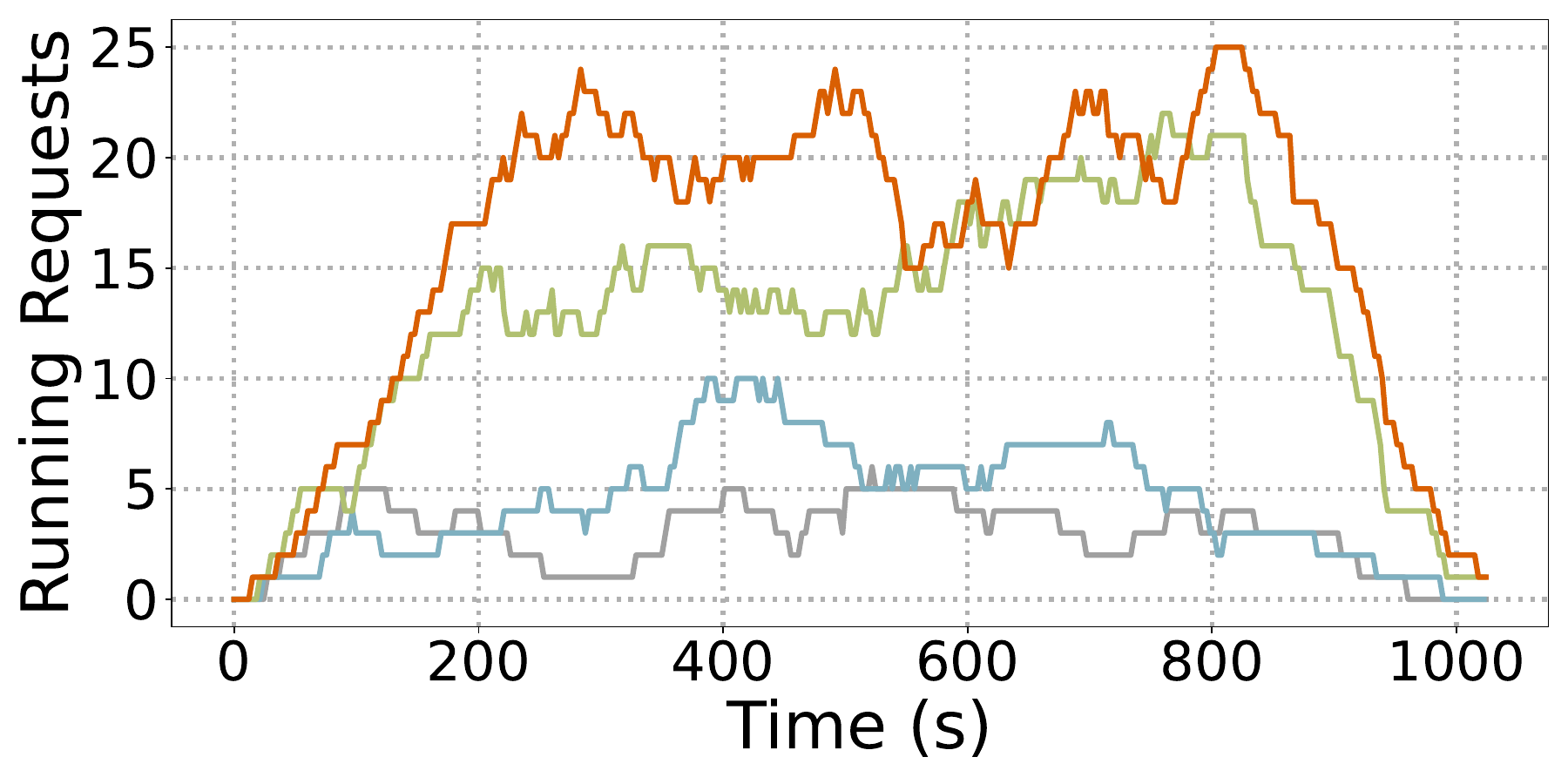}
        \caption{Stake amounts.}
        \label{fig:policy-stake}
    \end{subfigure}
    \hfill
    \begin{subfigure}{0.33\linewidth}
        \centering
        \includegraphics[width=\linewidth]{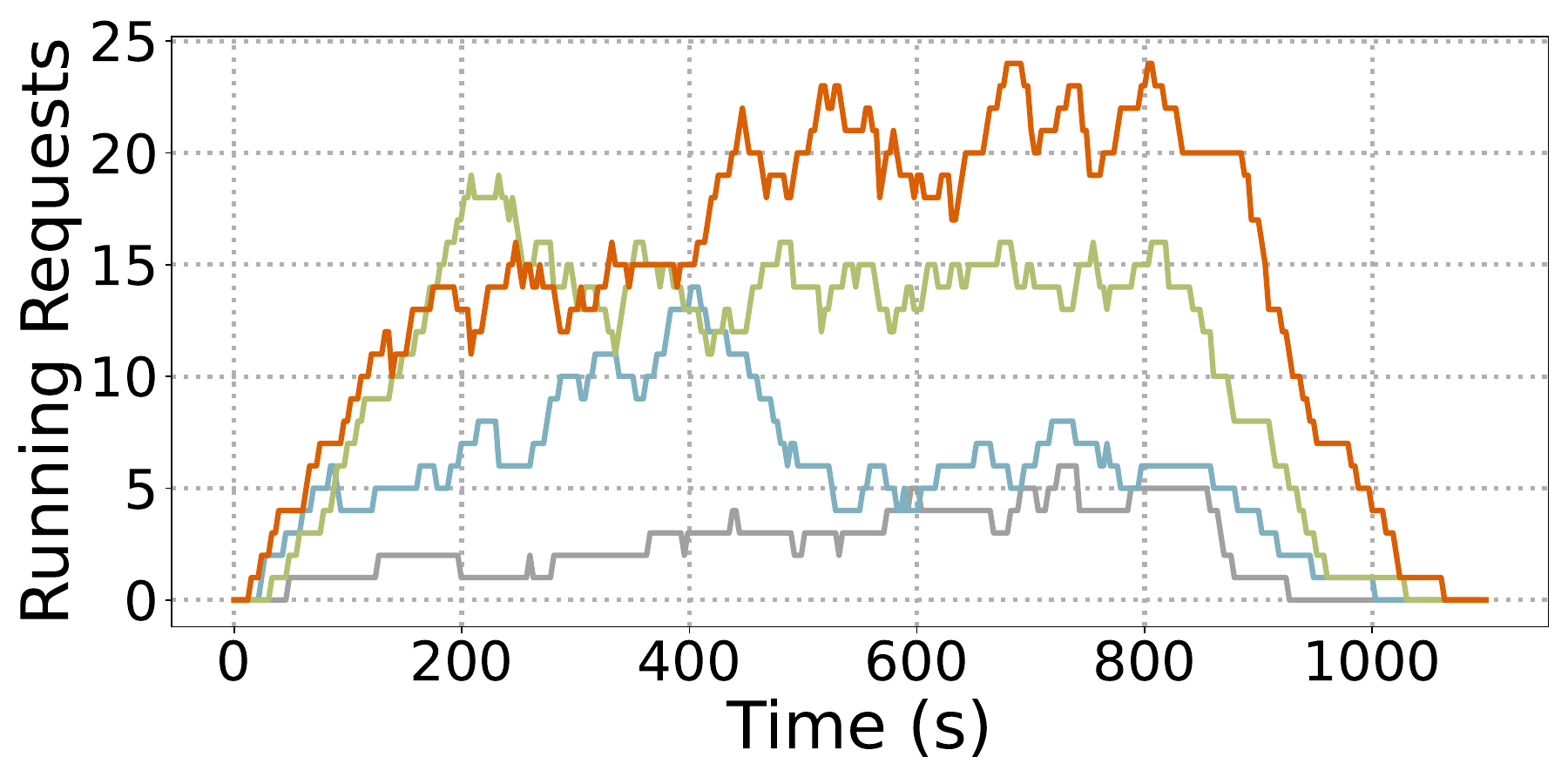}
        \caption{Acceptance frequencies.}
        \label{fig:policy-accept}
    \end{subfigure}
    \hfill
    \begin{subfigure}{0.32\linewidth}
        \centering
        \raisebox{0.2em}{
            \includegraphics[width=\linewidth]{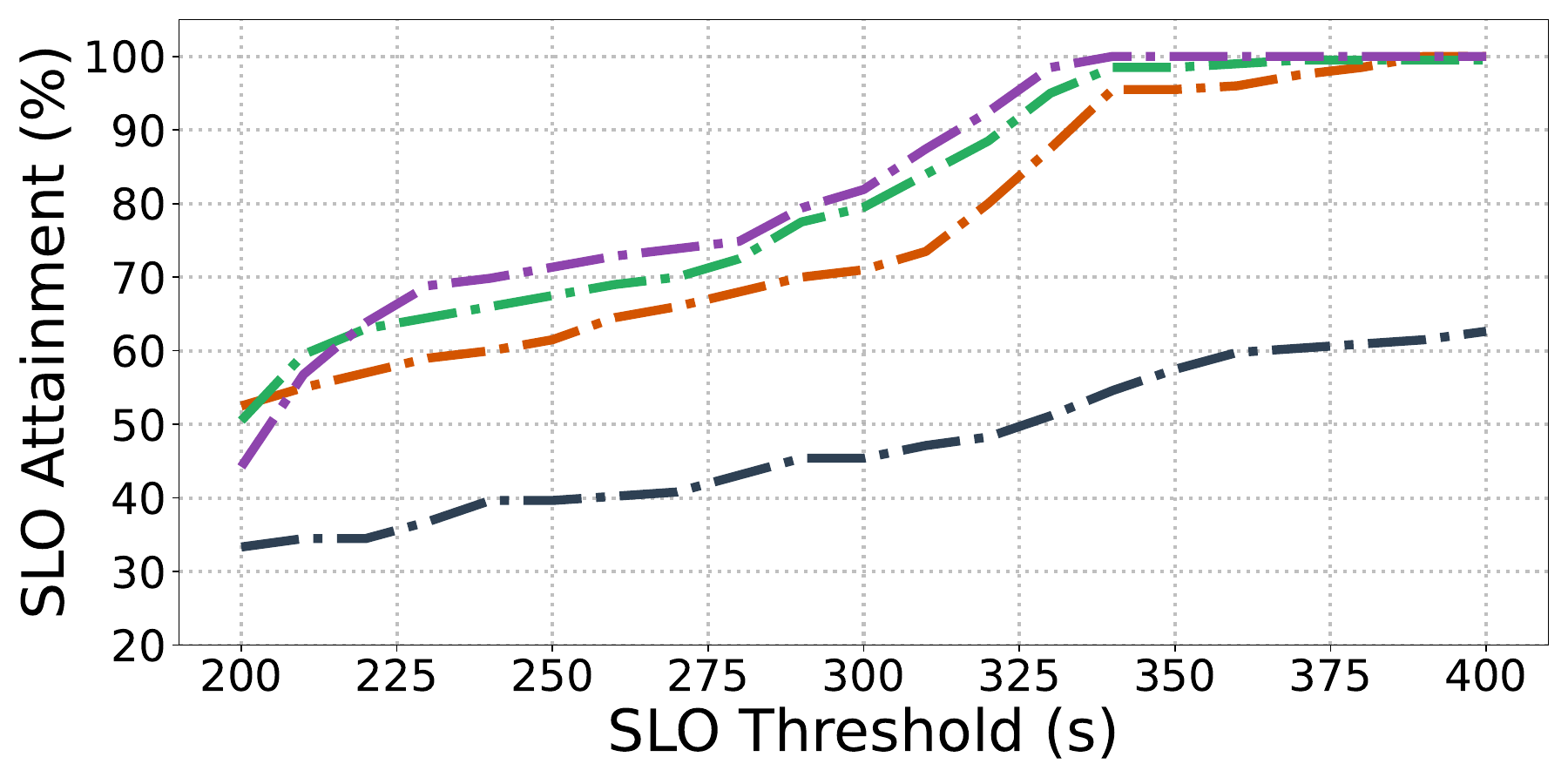}
        }
        \caption{Offloading frequencies.}
        \label{fig:policy-offload}
    \end{subfigure}
    \caption{Number of running requests under different stake amounts (left, $1, 2, 3, 4$ for Node 1-4) and different acceptance frequencies (middle, $0.25, 0.5, 0.75, 1.0$ for Node 1-4); SLO attainment under different offloading frequencies (right, $0.25, 0.5, 0.75, 1.0$ for Node 1-4).}
\end{figure}

We conduct an ablation to examine how user-level parameters (stake amount, request acceptance, and offloading frequency) affect workload allocation and global SLO attainment.

We first varied stake amounts and acceptance frequencies across nodes and monitored their local request queues. Requests were uniformly issued by a dedicated requester-only node. Figure~\ref{fig:policy-stake} and \ref{fig:policy-accept} show that nodes with higher stake or higher acceptance frequency handle a larger share of delegated requests. This demonstrates that the PoS-based scheduling faithfully reflects user-level policies, allowing nodes to actively control their participation. Next, we evaluated the effect of offloading frequency under sustained high request pressure. As illustrated in Figure~\ref{fig:policy-offload}, increasing offloading improves SLO attainment by redistributing workloads from overloaded nodes. However, the benefit saturates at moderate offloading rates: the improvement between rates of $50\%$, $75\%$, and $100\%$ is marginal. Excessive offloading can even hinder long-term credit accumulation as nodes spend more credits to delegate requests. Overall, these results confirm that WWW.Serve’s flexible policy framework allows service providers to regulate their participation and optimize both efficiency and credit dynamics, indicating substantial room for fine-tuning policies to better balance immediate performance and long-term incentives.

\section{Conclusion}
\label{sec:conclusion}

This paper presents WWW.Serve, a fully decentralized framework for collaborative LLM serving. Operating as an open, competitive market of global LLM services, our framework enables service providers to flexibly determine their participation policies and resource commitments, supports autonomous request routing and workload balancing, and incentivizes the delivery of high-quality LLM services. Experimental results show that WWW.Serve achieves scheduling efficiency comparable to centralized baselines while adapting effectively to dynamic resource availability, highlighting its potential as a scalable and privacy-preserving foundation for next-generation LLM services.

\section*{Acknowledgements}

We gratefully acknowledge access to NVIDIA computing resources. This work was partially supported by Google Research Award, Google ML \& System Junior Faculty Award, Amazon Research Award, Fireworks AI, Intel, Li Auto, Moffett AI, and CMU CyLab Seed funding. This research is also based upon work supported by the National Science Foundation under Grant Nos. CCF-2504353 and CCF-2247014, and by IARPA. Any opinions, findings, conclusions or recommendations expressed are those of the authors and do not necessarily reflect the views of the National Science Foundation.

\clearpage

\bibliographystyle{assets/plainnat}
\bibliography{paper}

@misc{openai,
  title        = {ChatGPT},
  author       = {OpenAI},
  howpublished = {\url{https://openai.com/chatgpt}},
  year         = {2022},
}

@misc{anthropic,
  title        = {Claude},
  author       = {Anthropic},
  howpublished = {\url{https://www.anthropic.com/claude}},
  year         = {2023},
}

@misc{azure,
  title        = {Azure OpenAI Service},
  author       = {Microsoft},
  howpublished = {\url{https://azure.microsoft.com/en-us/products/cognitive-services/openai-service/}},
  year         = {2023},
}

@misc{sglang,
    title={SGLang: Efficient Execution of Structured Language Model Programs}, 
    author={Lianmin Zheng and Liangsheng Yin and Zhiqiang Xie and Chuyue Sun and Jeff Huang and Cody Hao Yu and Shiyi Cao and Christos Kozyrakis and Ion Stoica and Joseph E. Gonzalez and Clark Barrett and Ying Sheng},
    year={2024},
    eprint={2312.07104},
    archivePrefix={arXiv},
    primaryClass={cs.AI},
    url={https://arxiv.org/abs/2312.07104}, 
}

@inproceedings{vllm,
    title={Efficient Memory Management for Large Language Model Serving with PagedAttention},
    author={Woosuk Kwon and Zhuohan Li and Siyuan Zhuang and Ying Sheng and Lianmin Zheng and Cody Hao Yu and Joseph E. Gonzalez and Hao Zhang and Ion Stoica},
    booktitle={Proceedings of the ACM SIGOPS 29th Symposium on Operating Systems Principles},
    year={2023}
}

@misc{helix,
    title={Helix: Serving Large Language Models over Heterogeneous GPUs and Network via Max-Flow}, 
    author={Yixuan Mei and Yonghao Zhuang and Xupeng Miao and Juncheng Yang and Zhihao Jia and Rashmi Vinayak},
    year={2025},
    eprint={2406.01566},
    archivePrefix={arXiv},
    primaryClass={cs.DC},
    url={https://arxiv.org/abs/2406.01566}, 
}

@misc{hexgen,
    title={HexGen: Generative Inference of Large Language Model over Heterogeneous Environment}, 
    author={Youhe Jiang and Ran Yan and Xiaozhe Yao and Yang Zhou and Beidi Chen and Binhang Yuan},
    year={2024},
    eprint={2311.11514},
    archivePrefix={arXiv},
    primaryClass={cs.DC},
    url={https://arxiv.org/abs/2311.11514}, 
}

@misc{chatbot-arena,
      title={Chatbot Arena: An Open Platform for Evaluating LLMs by Human Preference}, 
      author={Wei-Lin Chiang and Lianmin Zheng and Ying Sheng and Anastasios Nikolas Angelopoulos and Tianle Li and Dacheng Li and Hao Zhang and Banghua Zhu and Michael Jordan and Joseph E. Gonzalez and Ion Stoica},
      year={2024},
      eprint={2403.04132},
      archivePrefix={arXiv},
      primaryClass={cs.AI},
      url={https://arxiv.org/abs/2403.04132}, 
}

@misc{mtbench,
      title={Judging LLM-as-a-Judge with MT-Bench and Chatbot Arena}, 
      author={Lianmin Zheng and Wei-Lin Chiang and Ying Sheng and Siyuan Zhuang and Zhanghao Wu and Yonghao Zhuang and Zi Lin and Zhuohan Li and Dacheng Li and Eric P. Xing and Hao Zhang and Joseph E. Gonzalez and Ion Stoica},
      year={2023},
      eprint={2306.05685},
      archivePrefix={arXiv},
      primaryClass={cs.CL},
      url={https://arxiv.org/abs/2306.05685}, 
}

@misc{petals,
      title={Petals: Collaborative Inference and Fine-tuning of Large Models}, 
      author={Alexander Borzunov and Dmitry Baranchuk and Tim Dettmers and Max Ryabinin and Younes Belkada and Artem Chumachenko and Pavel Samygin and Colin Raffel},
      year={2023},
      eprint={2209.01188},
      archivePrefix={arXiv},
      primaryClass={cs.LG},
      url={https://arxiv.org/abs/2209.01188}, 
}

@misc{boinc,
      title={BOINC: A Platform for Volunteer Computing}, 
      author={David P. Anderson},
      year={2019},
      eprint={1903.01699},
      archivePrefix={arXiv},
      primaryClass={cs.DC},
      url={https://arxiv.org/abs/1903.01699}, 
}

@misc{ethereum,
      title={Unveiling Decentralization: A Comprehensive Review of Technologies, Comparison, Challenges in Bitcoin, Ethereum, and Solana Blockchain}, 
      author={Han Song and Yihao Wei and Zhongche Qu and Weihan Wang},
      year={2024},
      eprint={2404.04841},
      archivePrefix={arXiv},
      primaryClass={cs.CR},
      url={https://arxiv.org/abs/2404.04841}, 
}

@article{foldingathome,
      title={Folding@home: Achievements from over 20 years of distributed computing},
      author={Shirts, Michael R. and Pande, Vijay S.},
      journal={Current Opinion in Structural Biology},
      volume={80},
      pages={102569},
      year={2023},
      publisher={Elsevier},
      doi={10.1016/j.sbi.2023.102569},
      url={https://www.sciencedirect.com/science/article/pii/S0959440X23000745}
}

@misc{filecoin,
    title={Filecoin: A Decentralized Storage Network},
    author={Protocol Labs},
    howpublished={\url{https://filecoin.io/filecoin.pdf}},
    year={2017},
}

@misc{golem,
      title={Golem: Decentralized Supercomputing for Distributed Applications},
      author={Golem Network},
      howpublished={\url{https://assets.website-files.com/60005e3965a10f31d245af87/60352707e6dd742743c75764_Golemwhitepaper.pdf}},
      year={2020},
}

@misc{distserve,
      title={DistServe: Disaggregating Prefill and Decoding for Goodput-optimized Large Language Model Serving}, 
      author={Yinmin Zhong and Shengyu Liu and Junda Chen and Jianbo Hu and Yibo Zhu and Xuanzhe Liu and Xin Jin and Hao Zhang},
      year={2024},
      eprint={2401.09670},
      archivePrefix={arXiv},
      primaryClass={cs.DC},
      url={https://arxiv.org/abs/2401.09670}, 
}

@misc{gentorrent,
      title={GenTorrent: Scaling Large Language Model Serving with An Overlay Network}, 
      author={Fei Fang and Yifan Hua and Shengze Wang and Ruilin Zhou and Yi Liu and Chen Qian and Xiaoxue Zhang},
      year={2025},
      eprint={2504.20101},
      archivePrefix={arXiv},
      primaryClass={cs.DC},
      url={https://arxiv.org/abs/2504.20101}, 
}

@misc{deserve,
      title={DeServe: Towards Affordable Offline LLM Inference via Decentralization}, 
      author={Linyu Wu and Xiaoyuan Liu and Tianneng Shi and Zhe Ye and Dawn Song},
      year={2025},
      eprint={2501.14784},
      archivePrefix={arXiv},
      primaryClass={cs.DC},
      url={https://arxiv.org/abs/2501.14784}, 
}

@article{qwen3,
    title={Qwen3 Technical Report}, 
    author={An Yang and Anfeng Li and Baosong Yang and Beichen Zhang and Binyuan Hui and others},
    journal = {arXiv preprint arXiv:2505.09388},
    year={2025}
}

@misc{deepseekr1,
      title={DeepSeek-R1: Incentivizing Reasoning Capability in LLMs via Reinforcement Learning}, 
      author={DeepSeek-AI},
      year={2025},
      eprint={2501.12948},
      archivePrefix={arXiv},
      primaryClass={cs.CL},
      url={https://arxiv.org/abs/2501.12948}, 
}

@misc{llmchain,
      title={LLMChain: Blockchain-based Reputation System for Sharing and Evaluating Large Language Models}, 
      author={Mouhamed Amine Bouchiha and Quentin Telnoff and Souhail Bakkali and Ronan Champagnat and Mourad Rabah and Mickaël Coustaty and Yacine Ghamri-Doudane},
      year={2024},
      eprint={2404.13236},
      archivePrefix={arXiv},
      primaryClass={cs.DC},
      url={https://arxiv.org/abs/2404.13236}, 
}

@misc{bitcoin,
    author={Satoshi Nakamoto},
    title={Bitcoin: A Peer-to-Peer Electronic Cash System},
    year={2008},
    howpublished={\url{https://bitcoin.org/bitcoin.pdf}}
}

@inproceedings{ouroboros,
  author    = {Aggelos Kiayias and Alexander Russell and Bernardo David and Roman Oliynykov},
  title     = {Ouroboros: A Provably Secure Proof-of-Stake Blockchain Protocol},
  booktitle = {Advances in Cryptology – CRYPTO 2017},
  year      = {2017},
  pages     = {357--388},
  publisher = {Springer},
  doi       = {10.1007/978-3-319-63688-7_12}
}

@misc{casper,
      title={Casper the Friendly Finality Gadget}, 
      author={Vitalik Buterin and Virgil Griffith},
      year={2019},
      eprint={1710.09437},
      archivePrefix={arXiv},
      primaryClass={cs.CR},
      url={https://arxiv.org/abs/1710.09437}, 
}

@book{grid,
      editor    = {Ian Foster and Carl Kesselman},
      title     = {The Grid 2: Blueprint for a New Computing Infrastructure},
      publisher = {Morgan Kaufmann},
      year      = {2003},
      isbn      = {978-1558609334}
}

@misc{PoQ,
      title={Proof of Quality: A Costless Paradigm for Trustless Generative AI Model Inference on Blockchains}, 
      author={Zhenjie Zhang and Yuyang Rao and Hao Xiao and Xiaokui Xiao and Yin Yang},
      year={2024},
      eprint={2405.17934},
      archivePrefix={arXiv},
      primaryClass={cs.AI},
      url={https://arxiv.org/abs/2405.17934}, 
}

@misc{pariksha,
      title={PARIKSHA: A Large-Scale Investigation of Human-LLM Evaluator Agreement on Multilingual and Multi-Cultural Data}, 
      author={Ishaan Watts and Varun Gumma and Aditya Yadavalli and Vivek Seshadri and Manohar Swaminathan and Sunayana Sitaram},
      year={2024},
      eprint={2406.15053},
      archivePrefix={arXiv},
      primaryClass={cs.CL},
      url={https://arxiv.org/abs/2406.15053}, 
}

@misc{llama3,
      title={The Llama 3 Herd of Models}, 
      author={Touvron, Hugo and Yang, Luyu and Zhai, Shoufa and Pu, Tao and Lu, Zihang and others},
      year={2024},
      eprint={2407.21783},
      archivePrefix={arXiv},
      primaryClass={cs.AI},
      url={https://arxiv.org/abs/2407.21783}, 
}

@article{blockchainreview,
title = {A comprehensive review of blockchain technology: Underlying principles and historical background with future challenges},
journal = {Decision Analytics Journal},
volume = {9},
pages = {100344},
year = {2023},
issn = {2772-6622},
doi = {https://doi.org/10.1016/j.dajour.2023.100344},
url = {https://www.sciencedirect.com/science/article/pii/S2772662223001844},
author = {Gautami Tripathi and Mohd Abdul Ahad and Gabriella Casalino},
}

@misc{blockchainconsensus,
      title={Blockchain Consensus Protocols in the Wild}, 
      author={Christian Cachin and Marko Vukolić},
      year={2017},
      eprint={1707.01873},
      archivePrefix={arXiv},
      primaryClass={cs.DC},
      url={https://arxiv.org/abs/1707.01873}, 
}

@misc{consensusage,
      title={Consensus in the Age of Blockchains}, 
      author={Shehar Bano and Alberto Sonnino and Mustafa Al-Bassam and Sarah Azouvi and Patrick McCorry and Sarah Meiklejohn and George Danezis},
      year={2017},
      eprint={1711.03936},
      archivePrefix={arXiv},
      primaryClass={cs.CR},
      url={https://arxiv.org/abs/1711.03936}, 
}

@inproceedings{specinfer,
    series={ASPLOS ’24},
    title={SpecInfer: Accelerating Large Language Model Serving with Tree-based Speculative Inference and Verification},
    url={http://dx.doi.org/10.1145/3620666.3651335},
    DOI={10.1145/3620666.3651335},
    booktitle={Proceedings of the 29th ACM International Conference on Architectural Support for Programming Languages and Operating Systems, Volume 3},
    publisher={ACM},
    author={Miao, Xupeng and Oliaro, Gabriele and Zhang, Zhihao and Cheng, Xinhao and Wang, Zeyu and Zhang, Zhengxin and Wong, Rae Ying Yee and Zhu, Alan and Yang, Lijie and Shi, Xiaoxiang and Shi, Chunan and Chen, Zhuoming and Arfeen, Daiyaan and Abhyankar, Reyna and Jia, Zhihao},
    year={2024},
    month=apr, pages={932–949},
    collection={ASPLOS ’24}
}

@misc{proxy_ssjf,
      title={Efficient Interactive LLM Serving with Proxy Model-based Sequence Length Prediction}, 
      author={Haoran Qiu and Weichao Mao and Archit Patke and Shengkun Cui and Saurabh Jha and Chen Wang and Hubertus Franke and Zbigniew T. Kalbarczyk and Tamer Başar and Ravishankar K. Iyer},
      year={2024},
      eprint={2404.08509},
      archivePrefix={arXiv},
      primaryClass={cs.DC},
      url={https://arxiv.org/abs/2404.08509}, 
}

@inproceedings{flock,
author = {Chen, Ruonan and Dong, Ye and Liu, Yizhong and Fan, Tingyu and Li, Dawei and Guan, Zhenyu and Liu, Jianwei and Zhou, Jianying},
title = {FLock: Robust and Privacy-Preserving Federated Learning based on Practical Blockchain State Channels},
year = {2025},
isbn = {9798400712746},
publisher = {Association for Computing Machinery},
address = {New York, NY, USA},
url = {https://doi.org/10.1145/3696410.3714666},
doi = {10.1145/3696410.3714666},
booktitle = {Proceedings of the ACM on Web Conference 2025},
pages = {884–895},
numpages = {12},
location = {Sydney NSW, Australia},
series = {WWW '25}
}

@misc{fedshield,
      title={FedShield-LLM: A Secure and Scalable Federated Fine-Tuned Large Language Model}, 
      author={Md Jueal Mia and M. Hadi Amini},
      year={2025},
      eprint={2506.05640},
      archivePrefix={arXiv},
      primaryClass={cs.CR},
      url={https://arxiv.org/abs/2506.05640}, 
}

@misc{llm_blockchain,
      title={Connecting Large Language Models with Blockchain: Advancing the Evolution of Smart Contracts from Automation to Intelligence}, 
      author={Youquan Xian and Xueying Zeng and Duancheng Xuan and Danping Yang and Chunpei Li and Peng Fan and Peng Liu},
      year={2024},
      eprint={2412.02263},
      archivePrefix={arXiv},
      primaryClass={cs.DC},
      url={https://arxiv.org/abs/2412.02263}, 
}

@INPROCEEDINGS{linguachain,
  author={Kozgunov, Nikita V. and Khalashi, Mohammad Hossein and Oliseenko, Valerij D. and Tulupyeva, Tat’Jana V.},
  booktitle={2024 XXVII International Conference on Soft Computing and Measurements (SCM)}, 
  title={LinguaChain: a Peer-to-peer Dynamic Decentralized Large Language Model with Coin-based Incentives}, 
  year={2024},
  volume={},
  number={},
  pages={178-181},
  doi={10.1109/SCM62608.2024.10554241}
}

@misc{speculative1,
      title={Accelerating Large Language Model Decoding with Speculative Sampling}, 
      author={Charlie Chen and Sebastian Borgeaud and Geoffrey Irving and Jean-Baptiste Lespiau and Laurent Sifre and John Jumper},
      year={2023},
      eprint={2302.01318},
      archivePrefix={arXiv},
      primaryClass={cs.CL},
      url={https://arxiv.org/abs/2302.01318}, 
}

@misc{speculative2,
      title={Fast Inference from Transformers via Speculative Decoding}, 
      author={Yaniv Leviathan and Matan Kalman and Yossi Matias},
      year={2023},
      eprint={2211.17192},
      archivePrefix={arXiv},
      primaryClass={cs.LG},
      url={https://arxiv.org/abs/2211.17192}, 
}

@article{seti,
    author = {Anderson, David P. and Cobb, Jeff and Korpela, Eric and Lebofsky, Matt and Werthimer, Dan},
    title = {SETI@home: an experiment in public-resource computing},
    year = {2002},
    issue_date = {November 2002},
    publisher = {Association for Computing Machinery},
    address = {New York, NY, USA},
    volume = {45},
    number = {11},
    issn = {0001-0782},
    url = {https://doi.org/10.1145/581571.581573},
    doi = {10.1145/581571.581573},
    journal = {Commun. ACM},
    month = nov,
    pages = {56–61},
    numpages = {6}
}

@misc{survey_dec_training,
      title={Beyond A Single AI Cluster: A Survey of Decentralized LLM Training}, 
      author={Haotian Dong and Jingyan Jiang and Rongwei Lu and Jiajun Luo and Jiajun Song and Bowen Li and Ying Shen and Zhi Wang},
      year={2025},
      eprint={2503.11023},
      archivePrefix={arXiv},
      primaryClass={cs.DC},
      url={https://arxiv.org/abs/2503.11023}, 
}

@misc{survey_dec_learning,
      title={A survey on secure decentralized optimization and learning}, 
      author={Changxin Liu and Nicola Bastianello and Wei Huo and Yang Shi and Karl H. Johansson},
      year={2024},
      eprint={2408.08628},
      archivePrefix={arXiv},
      primaryClass={cs.LG},
      url={https://arxiv.org/abs/2408.08628}, 
}

@article{think_p2p,
  author       = {Anne{-}Marie Kermarrec and
                  Fran{\c{c}}ois Ta{\"{\i}}ani},
  title        = {Want to scale in centralized systems? Think {P2P}},
  journal      = {J. Internet Serv. Appl.},
  volume       = {6},
  number       = {1},
  pages        = {16:1--16:12},
  year         = {2015},
  url          = {https://doi.org/10.1186/s13174-015-0029-1},
  doi          = {10.1186/S13174-015-0029-1},
  bibsource    = {dblp computer science bibliography, https://dblp.org}
}

@misc{dec_privacy_preserving,
      title={Decentralized Privacy-preserving Timed Execution in Blockchain-based Smart Contract Platforms}, 
      author={Chao Li and Balaji Palanisamy},
      year={2019},
      eprint={1902.05613},
      archivePrefix={arXiv},
      primaryClass={cs.CR},
      url={https://arxiv.org/abs/1902.05613}, 
}

@inproceedings{byzantine,
author = {Ma, Kehao and Xu, Minghui and Guo, Yihao and Cui, Lukai and Ni, Shiping and Zhang, Shan and Wang, Weibing and Yang, Haiyong and Cheng, Xiuzhen},
title = {Anonymity on Byzantine-Resilient Decentralized Computing},
year = {2024},
isbn = {978-3-031-71466-5},
publisher = {Springer-Verlag},
address = {Berlin, Heidelberg},
url = {https://doi.org/10.1007/978-3-031-71467-2_32},
doi = {10.1007/978-3-031-71467-2_32},
booktitle = {Wireless Artificial Intelligent Computing Systems and Applications: 18th International Conference, WASA 2024, Qindao, China, June 21–23, 2024, Proceedings, Part II},
pages = {400–412},
numpages = {13},
location = {Qingdao, China}
}

@misc{eccos,
      title={OmniRouter: Budget and Performance Controllable Multi-LLM Routing}, 
      author={Kai Mei and Wujiang Xu and Shuhang Lin and Yongfeng Zhang},
      year={2025},
      eprint={2502.20576},
      archivePrefix={arXiv},
      primaryClass={cs.DB},
      url={https://arxiv.org/abs/2502.20576}, 
}

@misc{openr1math220k,
  author = {Open-R1-Team},
  title = {OpenR1-Math-220k: A Large-Scale Dataset for Mathematical Reasoning},
  year = {2025},
  publisher = {Hugging Face},
  url = {https://huggingface.co/datasets/open-r1/OpenR1-Math-220k},
}

@misc{torchao,
  title={TorchAO: PyTorch-Native Training-to-Serving Model Optimization},
  author={torchao},
  url={https://github.com/pytorch/ao},
  license={BSD-3-Clause},
  month={oct},
  year={2024}
}

@article{flashinfer,
    title = {FlashInfer: Efficient and Customizable Attention Engine for LLM Inference Serving},
    author = {
      Ye, Zihao and
      Chen, Lequn and
      Lai, Ruihang and
      Lin, Wuwei and
      Zhang, Yineng and
      Wang, Stephanie and
      Chen, Tianqi and
      Kasikci, Baris and
      Grover, Vinod and
      Krishnamurthy, Arvind and
      Ceze, Luis
    },
    journal = {arXiv preprint arXiv:2501.01005},
    year = {2025},
    url = {https://arxiv.org/abs/2501.01005}
}

@misc{sdpa,
      title={Attention Is All You Need}, 
      author={Ashish Vaswani and Noam Shazeer and Niki Parmar and Jakob Uszkoreit and Llion Jones and Aidan N. Gomez and Lukasz Kaiser and Illia Polosukhin},
      year={2023},
      eprint={1706.03762},
      archivePrefix={arXiv},
      primaryClass={cs.CL},
      url={https://arxiv.org/abs/1706.03762}, 
}

@misc{triton,
  title        = {Triton},
  author       = {OpenAI},
  url          = {https://github.com/triton-lang/triton},
  year         = {2021},
}

\clearpage

\beginappendix
\section{Supplementary System Details}

In this section, we provide additional illustrations that complement the descriptions in Section~\ref{sec:overview} and Section~\ref{sec:core_design}, focusing on two key components of WWW.Serve: (i) the end-to-end workflow of processing a single user request, and (ii) the gossip-driven protocol for peer synchronization.

\subsection{Request Processing Workflow}

\begin{figure}[t]
    \centering
    \includegraphics[width=\linewidth]{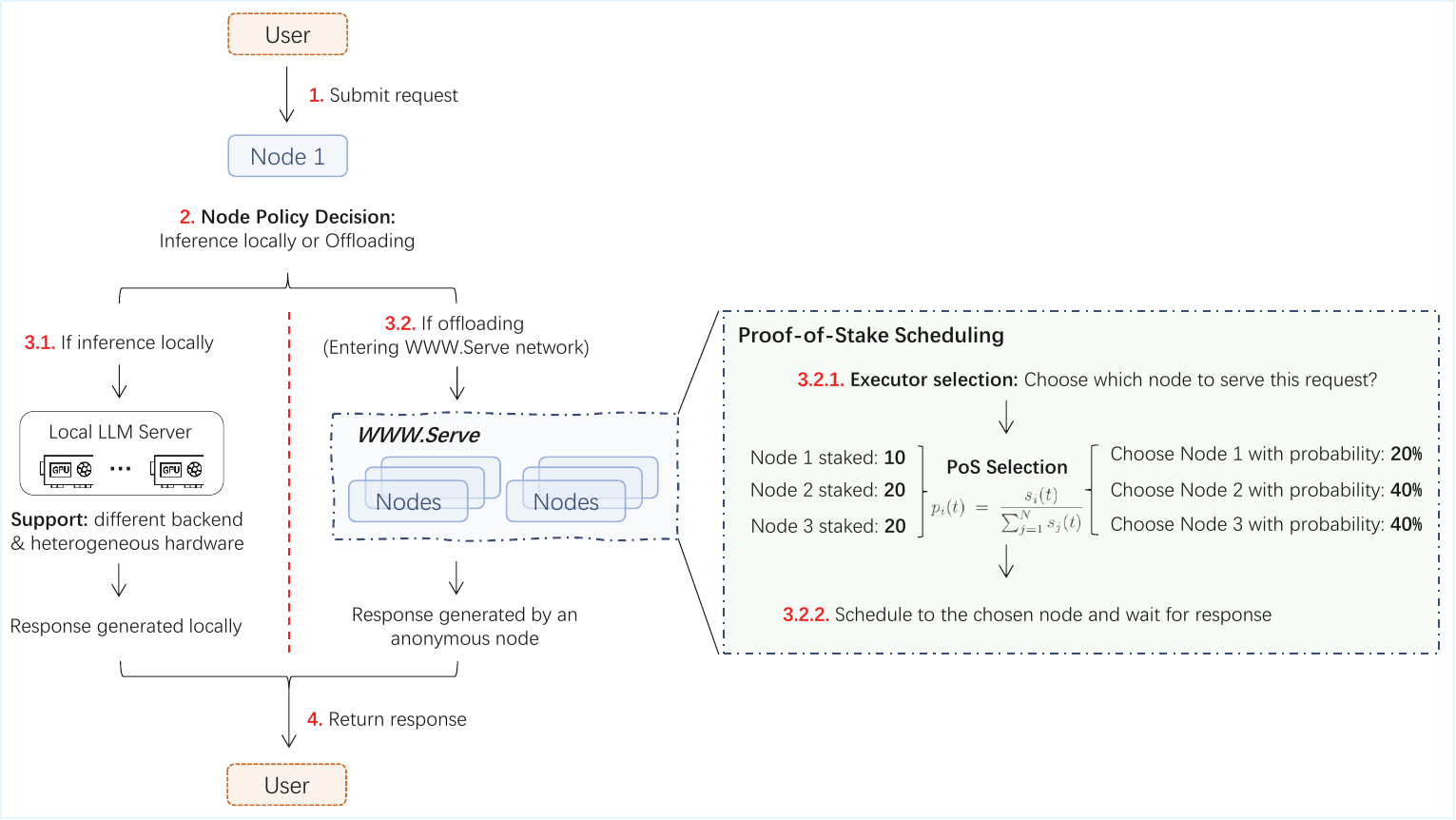}
    \caption{End-to-end workflow of a single user request, including local execution or remote offloading via PoS-based scheduling.}
    \label{fig:request_workflow}
\end{figure}

Figure~\ref{fig:request_workflow} presents the end-to-end workflow of a node handling a user request. Upon receiving a query (Step 1), the node determines whether to execute it locally or offload it to the network (Step 2).

\textbf{Local execution.} (Step 3.1) Nodes may host local language models using diverse runtimes (e.g., vLLM, SGLang) on heterogeneous devices. WWW.Serve abstracts these differences through a unified inference interface, allowing heterogeneous hardware and software stacks to participate without modifications to global collaboration mechanisms.

\textbf{Remote execution.} (Step 3.2) If offloading is selected, the node samples a trustworthy executor through our PoS-based scheduler (Step 3.2.1), where each peer’s sampling probability is proportional to its staked credit. Once an executor accepts the task, the request is forwarded for processing and the generated response is returned to the origin node.

\subsection{Gossip-Driven Peer Synchronization}

\begin{figure}[t]
    \centering
    \includegraphics[width=\linewidth]{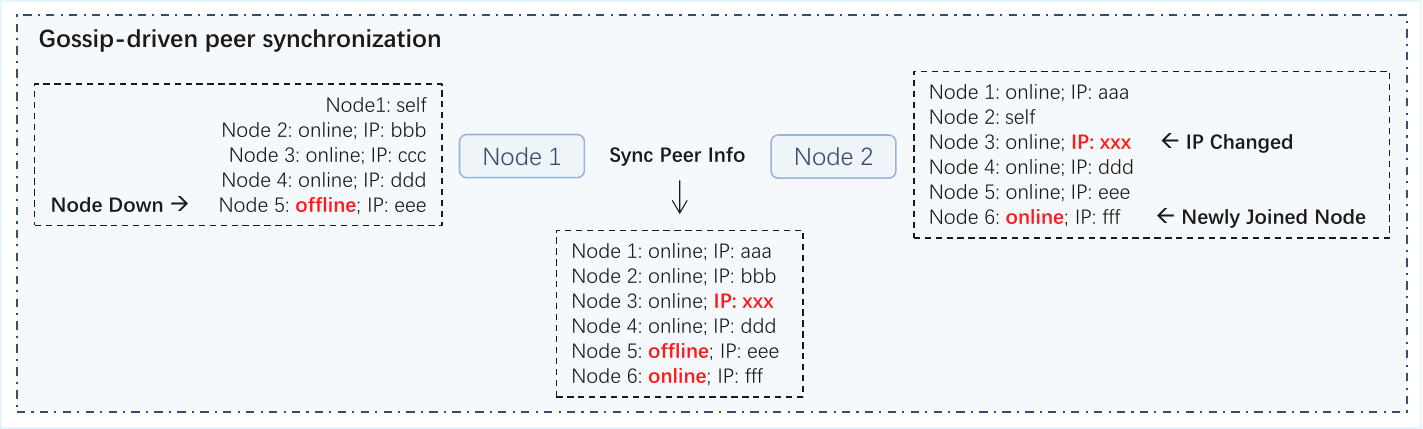}
    \caption{Gossip-driven peer synchronization. During each gossip round, nodes exchange local peer views, allowing updated information to propagate diffusively throughout the network.}
    \label{fig:gossip_detail}
\end{figure}

Figure~\ref{fig:gossip_detail} shows an example gossip synchronization between two nodes. Each node maintains a local view of peer availability, including identifiers, online/offline status, and communication endpoints. During a gossip round, two nodes exchange their current views and reconcile any discrepancies, for instance, peers that have gone offline (Node 5), updated their network addresses (Node 3), or newly joined (Node 6). Repeated lightweight pairwise exchanges allow updates to diffuse across the network and converge quickly, without requiring any central coordinator.

\section{Implementation}
\label{app:implementation}

In this section, we detail the implementation of the core modules contained in WWW.Serve.

\emph{Communication Manager:} is implemented using \texttt{ZeroMQ}, providing low-latency, asynchronous message passing between nodes. We adopt the \texttt{ROUTER} pattern, where each node binds to a fixed port to listen for incoming messages while simultaneously sending requests to peers. This design enables efficient bidirectional communication without relying on a centralized broker.

\emph{Request Manager:} leverages an asynchronous queue (\texttt{AsyncQueue}) for local request buffering and scheduling. Incoming requests are timestamped and inserted into the queue, while outgoing requests are dynamically dispatched to eligible executors based on the Proof-of-Stake–based selection mechanism and user-specific rules.

\emph{Model Manager:} supports a variety of LLM serving backends via \texttt{AsyncOpenAI} clients. Service providers only need to supply a base URL and API key, without exposing internal model details. Each node periodically collects metrics from its backend servers, including the number of active and queued requests and memory utilization, to support efficient request dispatching and balanced workload distribution.

\emph{Experiment Configuration:} is specified in a dedicated YAML file, capturing all necessary parameters for a node to initialize WWW.Serve modules. Each file includes: (i) Server Parameters: communication IP, port, user-level policy (e.g., stake, offload frequency, accept frequency), and backend selection (e.g., SGLang, vLLM); and (ii) Models: paths to local or remote LLMs, base URL for API access, and API keys. Each model entry also specifies generation parameters (e.g., maximum tokens, temperature, top-p) and dispatch parameters (e.g., target memory utilization). These YAML files are automatically parsed by each node at startup, ensuring reproducibility and allowing fine-grained control over node behavior.

\section{Experimental Settings}
\label{app:setting}

\begin{table}[t]
    \centering
    \small
    \renewcommand{\arraystretch}{1.2}
    \begin{tabular}{l|lll|cc|cc}
        \toprule
        \multirow{2}{*}{\textbf{Node}} & \multirow{2}{*}{\textbf{Model}} & \multirow{2}{*}{\textbf{GPU}} & \multirow{2}{*}{\textbf{Backend}} & \multicolumn{4}{c}{\textbf{Request Schedule}} \\
        & & & & Interval 1 & $1/\lambda_1$ & Interval 2 & $1/\lambda_2$ \\
        \midrule
        \rowcolor{gray!10}
        \multicolumn{8}{l}{\textit{\qquad Setting 1}} \\
        Node 1 & Qwen3 8B & ADA6000 & SGLang & 0--300s & 5 & 300--750s & 20 \\
        Node 2 & Qwen3 8B & ADA6000 & SGLang & 0--750s & 20 & & \\
        Node 3 & Qwen3 8B & ADA6000 & SGLang & 0--750s & 20 & & \\
        Node 4 & Qwen3 8B & ADA6000 & SGLang & 0--450s & 20 & 450--750s & 5 \\
        \midrule
        \rowcolor{gray!10}
        \multicolumn{8}{l}{\textit{\qquad Setting 2}} \\
        Node 1 & Qwen3 8B & ADA6000 & SGLang & 0--300s & 4 & 300--750s & 20 \\
        Node 2 & Qwen3 8B & ADA6000 & SGLang & 0--750s & 20 & & \\
        Node 3 & Qwen3 4B & RTX3090 & SGLang & 0--750s & 30 & & \\
        Node 4 & Qwen3 4B & RTX3090 & SGLang & 0--450s & 30 & 450--750s & 6 \\
        \midrule
        \rowcolor{gray!10}
        \multicolumn{8}{l}{\textit{\qquad Setting 3}} \\
        Node 1 & Qwen3 32B & 4$\times$A100 & SGLang & 0--300s & 2 & 300--750s & 6 \\
        Node 2 & Qwen3 8B & L40S & SGLang & 0--750s & 15 & & \\
        Node 3 & DeepSeek-Qwen 7B & RTX3090 & vLLM & 0--750s & 30 & & \\
        Node 4 & Llama3.1 8B & ADA6000 & vLLM & 0--450s & 15 & 450--750s & 5 \\
        \midrule
        \rowcolor{gray!10}
        \multicolumn{8}{l}{\textit{\qquad Setting 4}} \\
        Node 1 & Llama3.1 8B & L40S & vLLM & 0--750s & 9 & & \\
        Node 2 & Llama3.1 8B & L40S & vLLM & 0--450s & 6 & 450--750s & 12 \\
        Node 3 & DeepSeek-Qwen 7B & ADA6000 & vLLM & 0--300s & 6 & 300--750s & 12 \\
        Node 4 & DeepSeek-Qwen 7B & ADA6000 & vLLM & 0--450s & 12 & 450--750s & 6 \\
        Node 5 & Qwen3 4B & RTX4090 & SGLang & 0--750s & 12 & & \\
        Node 6 & Qwen3 4B & RTX4090 & SGLang & 0--450s & 10 & 450--750s & 20 \\
        Node 7 & Qwen3 4B & RTX3090 & SGLang & 0--300s & 20 & 300--750s & 10 \\
        Node 8 & Qwen3 4B & RTX3090 & SGLang & 0--300s & 20 & 300--750s & 10 \\
        \bottomrule
    \end{tabular}
    \caption{Experimental configurations correspond to Figure~\ref{fig:basic} (left to right) and Table~\ref{tab:latency}. Each setting specifies the deployed model, GPU type, serving backend, and the time-varying request schedule for all nodes. The Interval columns specify the time ranges, and the corresponding $1/\lambda$ columns denote the expected inter-arrival time (in seconds) used for Poisson request generation, i.e., request inter-arrival times distributed as $Poi(\lambda)$.}
    \label{tab:settings}
\end{table}

To comprehensively evaluate the scheduling efficiency of WWW.Serve in heterogeneous, dynamic environments, we designed four distinct experimental settings, summarized in Table~\ref{tab:settings}. Each setting varies in the deployed language models, GPU types, and serving backends, covering a broad spectrum of realistic node capabilities. Our evaluation primarily relies on recent open-source reasoning LLMs, including the Qwen3 series~\citep{qwen3}, DeepSeek-Qwen~\citep{deepseekr1}, and LLaMA~3.1~\citep{llama3}, and prompts are drawn from the OpenR1-Math-220k dataset~\citep{openr1math220k}. Time-varying request patterns are simulated via piecewise Poisson arrival rates for each node, capturing both high- and low-load periods that differ across nodes. Due to the limited scale of our experiments, we employ a shared ledger instead of a full Credit Block Chain, simplifying implementation while preserving the essential dynamics of credit transactions.

All nodes are configured with consistent policy parameters, including offload frequency (80\%), acceptance frequency (80\%), target utilization (70\%), and generation parameters such as maximum token length (8192), temperature (0), and top-p sampling (0.95). These standardized settings ensure comparability and reproducibility across heterogeneous nodes while enabling a systematic evaluation of the effects of resource diversity and dynamic workloads on scheduling efficiency, latency, and SLO attainment.

\end{document}